\documentclass[11pt, letterpaper, reqno]{amsart}
\usepackage[utf8]{inputenc} 	
\usepackage{microtype} 			
\usepackage[marginparwidth=1.2in]{geometry} 			
\usepackage{amsmath}			
\usepackage{amsthm}		 		
\usepackage{amssymb}	 		
\usepackage{bm}					
\usepackage{mathrsfs}			
\usepackage{mathtools}
\usepackage{xcolor}				
\definecolor{darkred}{rgb}{0.6, 0.1, 0.1}
\definecolor{darkblue}{rgb}{0.2, 0.2, 0.6}
\definecolor{darkgreen}{rgb}{0.2, 0.4,.1}
\definecolor{mellowyellow}{rgb}{1,.8,.2}
\definecolor{bettercyan}{rgb}{0.1, 0.4, 0.7}
\definecolor{bettermagenta}{rgb}{0.6, 0.2, 0.6}
\usepackage{tikz}				
\usetikzlibrary{patterns, arrows}
\usetikzlibrary{decorations.markings,decorations.pathmorphing,shapes}
\usepackage[all]{xy}			
\usepackage{graphicx}			
\usepackage{pgfplots}			
\usepackage{enumitem} 			
\usepackage{array}
\usepackage{perpage}	
\usepackage{lmodern}			
\usepackage[T1]{fontenc}		
\usepackage[breaklinks=true, colorlinks=true, citecolor=bettercyan, linkcolor=darkred, urlcolor=bettermagenta]{hyperref}
\usepackage{cleveref}
\usepackage[backend=biber, style=alphabetic, date=year, backref=true, giveninits=true, isbn=false]{biblatex}
\usepackage{comment}
\pgfplotsset{compat=1.15}

\makeatletter
\renewbibmacro{in:}{}
\DeclareFieldFormat{pages}{#1}
\renewcommand*{\bibnamedash}{%
	\leavevmode\raise +0.6ex\hbox to 5.5ex{\hrulefill}.\space\space}

\InitializeBibliographyStyle{\global\undef\bbx@lasthash}
\newbibmacro*{bbx:savehash}{\savefield{fullhash}{\bbx@lasthash}}

\renewbibmacro*{author}{%
	\ifboolexpr{
		test \ifuseauthor
		and
		not test {\ifnameundef{author}}
	}
	{%
		\iffieldequals{fullhash}{\bbx@lasthash}
		{\bibnamedash\addcomma\space}
		{\printnames{author}}%
		\usebibmacro{bbx:savehash}%
		\iffieldundef{authortype}
		{}
		{%
			\setunit{\addcomma\space}%
			\usebibmacro{authorstrg}%
		}%
	}
	{\global\undef\bbx@lasthash}%
}
\makeatother

\MakePerPage{footnote}

\newtheorem{propositionx}{Proposition}[section]
\newenvironment{proposition}
{\pushQED{\qed}\propositionx}
{\popQED\endpropositionx}

\newtheorem{theorem}{Theorem}

\newtheorem{corollaryx}[propositionx]{Corollary}

\newtheorem{lemmax}[propositionx]{Lemma}
\newenvironment{lemma}
{\pushQED{\qed}\lemmax}
{\popQED\endlemmax}

\theoremstyle{definition}
\newtheorem{definition}[propositionx]{Definition}
\newtheorem{example}[propositionx]{Example}

\theoremstyle{remark}

\newtheorem*{remark*}{Remark}

\newcommand{\bbC}{\mathbb{C}}

\newcommand{\bbN}{\mathbb{N}}

\newcommand{\bbR}{\mathbb{R}}
\newcommand{\bbS}{\mathbb{S}}

\newcommand{\calD}{\mathcal{D}}
\newcommand{\calE}{\mathcal{E}}
\newcommand{\calF}{\mathcal{F}}

\newcommand{\calL}{\mathcal{L}}

\newcommand{\calP}{\mathcal{P}}

\newcommand{\calR}{\mathcal{R}}
\newcommand{\calS}{\mathcal{S}}

\newcommand{\calV}{\mathcal{V}}

\newcommand{\calX}{\mathcal{X}}
\newcommand{\calY}{\mathcal{Y}}

\newcommand{\bfa}{\mathbf{a}}

\newcommand{\bfx}{\mathbf{x}}

\newcommand{\dd}{\,\mathrm{d}}

\newcommand{\supp}{\operatorname{supp}}

\newcommand{\WF}{\operatorname{WF}}
\DeclareMathOperator{\Op}{Op}

\newcommand{\RR}{\mathbb{R}}
\newcommand{\NN}{\mathbb{N}}
\newcommand{\loc}{\text{loc}}

\newcommand{\pa}{\partial}
\newcommand{\scb}{\mathrm{sc-b}}
\renewcommand{\sc}{\text{sc}}
\newcommand{\tF}{\tilde{F}}

\renewcommand{\Im}{\operatorname{Im}}
\renewcommand{\Re}{\operatorname{Re}}

\title[On the convergence of the Born series]{On the convergence of the Born series for Coulomb potentials}
\date{Last updated: August 7th, 2026}
\author{Ethan Sussman}
\author{Jared Wunsch}

\begin{document}
	
\begin{abstract}
	We provide a short proof of the convergence of the Born series on asymptotically conic manifolds, at sufficiently high energy. The potential is allowed to have multiple Coulomb singularities. This is handled using powerful semiclassical estimates recently proven by Hintz for the case of a single dipole (or better) singularity. The potential is also allowed to be long-range, like the actual Coulomb potential $1/r$; long-range potentials are handled using anisotropic semiclassical (sc-) Sobolev spaces. As a consequence of the above estimates, we show the existence of a resonance-free region for Hamiltonians with multiple screened Coulomb singularities.
\end{abstract}

\maketitle

\tableofcontents

\section{Introduction}

Let $X$ be a non-trapping asymptotically conic manifold \cite{Me94,MelroseGeometric}, with dimension $d\geq 3$. 
For example, the manifold-with-boundary $X$ could be the radial compactification of $\bbR^d$ (the ``exactly Euclidean'' case), 
\begin{equation}
    \overline{\bbR^d} = \bbR^d\sqcup \infty \bbS^{d-1},
\end{equation}
in which case $X^\circ = \bbR^d$; the reader may then forget $\partial X=\infty \bbS^{d-1}$, for most purposes. Part of the data of an asymptotically conic manifold is a Riemannian metric on the interior; in the exactly Euclidean case, this is the Euclidean metric.
Suppose we mark some finite number $\#\in \bbN$ of points, $p_1,\dots,p_\# \in X^\circ$.
Consider a possibly long-range potential with Coulomb singularities at those marked points, but smooth otherwise. That is,
\begin{equation}\label{eq:Coulomb}
	V\in r^{-1} C^\infty(X;\bbR), 
\end{equation}
where $r \in C^0(X^\circ ; \bbR^{\geq 0})$ is a function such that
\begin{itemize}
        \item $r^2 \in C^\infty(X^\circ;\bbR)$, so $r^2$ is smooth at each marked point (note that $X^\circ$ includes the marked points), 
        \item $r^2$ is a quadratic-defining-function of the union of the marked points, which means that 
        \begin{enumerate}[label=(\roman*)]
        \item for $x\in X^\circ$, we have $r(x)=0$ if and only if $x$ is one of our marked points,  
        \item the Hessian of $r^2$ is non-degenerate at each marked point,
        \end{enumerate}
	\item $\langle r \rangle^{-1} \in C^\infty(X)$ is a boundary-defining-function of $\partial X$.
\end{itemize}
(When $X=\overline{\bbR^d}$ and we are interested in a single marked point, then $r$ can be the distance to the marked point. However, our main interest is in situations with multiple cone points.)\footnote{Actually, in the body of this paper, we will allow (real-valued) $V \in r^{-1} S^0([X;\text{markings}])$, where, for any compact manifold-with-boundary $Y$, $S^0(Y)$ denotes the set of functions smooth on the interior which are conormal to the boundary. Here, $[X;\text{markings}]$ is the result of blowing up each marked point in $X$, replacing it with a $\bbS^{d-1}$. This effectively means passing to spherical coordinates near each marked point. Being conormal of weight 0 means remaining bounded under powers of the tangential vector fields $r\pa_r$ and $\pa_\theta$.}

In this note, we are interested in the high energy behavior of the
``limiting resolvent'' 
\begin{equation}
	R_V(\sigma \pm i0) = \operatorname{s-lim}_{\varepsilon \to 0^+} (\underbrace{\triangle + V}_P - \sigma^2 \mp i\varepsilon )^{-1}, \quad \sigma>0
\end{equation}
associated to the Schr\"odinger operator 
\begin{equation} 
P=\triangle+V.
\end{equation}
Here, $\triangle$ is the Laplace--Beltrami
operator on $X$ (the metric dependence being suppressed in the
notation) with the sign convention that $\triangle$ is positive semi-definite.  High energy means $\sigma\to\infty$; $\sigma$ is the ``frequency'' and $\sigma^2$ the ``energy.'' The various objects of interest in scattering theory (perturbed plane waves, the scattering matrix, the spectral measure, etc.)\ can be written in terms of the two limiting resolvents.

When the potential is small, one expects the methods of perturbation theory to be available, in particular, the \emph{Born} 
(a.k.a.\ Neumann) \emph{series}
\begin{equation}
R_V(\sigma \pm i0) \overset{?}{=} \sum_{j=0}^\infty R_0(\sigma \pm i0) (-M_V R_0(\sigma \pm i0))^j,
\end{equation}
relating $R_V$ to the free resolvent $R_0$. Here, $M_V:u\mapsto Vu$ is the usual multiplication operator. 
At issue is the convergence of this series --- whether it converges, and if so, whether it converges to $R_V$.
Indeed, when the potential is sufficiently small, 
\begin{equation}
    R_V(\sigma \pm i0) = \lim_{J\to\infty}R_0(\sigma \pm i0)  \sum_{j=0}^J (-M_V R_0(\sigma \pm i0))^j,
\end{equation}
in suitable operator-norm topologies.
In the Euclidean case, this is very classical.
Convergence is \emph{also} true if $\sigma$ is sufficiently large, regardless of whether or not $V$ is small. The purpose of this note is to prove this in a rather general setting.

The most elementary result regards uniform $L^2_{\mathrm{loc}}\to L^2_{\mathrm{loc}}$ bounds, in the 
$\sigma \to \infty$ limit, of the terms in the Born series.  Note that the
$j=0$ term is already well-studied \cite{VaZw00} and enjoys the estimate
\begin{equation} \label{eq:sharp_0}
\lVert M_\chi R_0(\sigma\pm i 0) M_\chi \rVert_{L^2 \to
  L^2}=O\Big(\frac{1}{\sigma}\Big).
\end{equation}
 
For shorthand, let
\begin{equation}
B_j(\sigma \pm i0) = R_0(\sigma\pm i0)  (-M_V R_0(\sigma\pm i0))^j.
\label{eq:Born_series}
\end{equation} 
denote the $j$th term in the Born series, $\mathsf{B}(\sigma \pm i0) = \sum_{j=0}^\infty B_j(\sigma \pm i0)$.\footnote{We use the sans-serif font `$\mathsf{B}$' to denote that the Born series should be considered as a formal series until proven otherwise.}
We prove:
\begin{theorem}\label{thm:main_0}
For all $\chi\in C_{\mathrm{c}}^\infty(X^\circ)$ and $\varepsilon,\Sigma>0$,  
there exists a constant $C>0$ such that the $j$th term $B_j(\sigma \pm i0)$ in the Born series obeys a bound 
\begin{equation}
\lVert M_\chi B_j(\sigma \pm i0) M_\chi \rVert_{L^2\to L^2} \leq \Big(\frac{C}{\sigma^{1-\varepsilon}} \Big)^{j+1}
\label{eq:uniform}
\end{equation}
for all $\sigma>\Sigma$. Consequently, if $\sigma$ is large enough,
the cutoff Born series $M_\chi\mathsf{B}(\sigma \pm i0)M_\chi$ converges geometrically in the $L^2\to L^2$ operator norm topology
to $M_\chi R_V(\sigma \pm i0) M_\chi$, and the estimate 
\begin{equation}\label{eq:sharp_full}
		\lVert M_\chi R_V(\sigma \pm i0) M_\chi \rVert_{L^2\to L^2} = O\Big(\frac{1}{\sigma} \Big)\quad\;\; 
\end{equation}
holds, just as in the free case.\footnote{The $O(1/\sigma)$ estimate \cref{eq:sharp_full} is proven by combining the sharp $j=0$ bound \cref{eq:sharp_0} with the $\varepsilon$-lossy \cref{eq:uniform} for $j\geq 1$. Using the latter for $j=0$ results in an $O(1/\sigma^{1-\varepsilon})$ estimate instead.}
\end{theorem}

\begin{remark*}
	More refined estimates, including control on derivatives, and without the cutoffs $\chi$, are proven below. 
	Although we restrict attention to $L^2$-based Sobolev spaces here, estimates in other $L^p$ spaces follow via Sobolev embedding.
\end{remark*} 

Because $V$ has simple poles at the marked points, it might look like the $j$th term in the Born series has a severe singularity. Part of what needs to be shown is that the free resolvent $R_0(\sigma \pm i0)$ mollifies the singularity by at least one $O(r)$ each time it is applied. Hardy's inequality justifies this easily, but at an $O(\sigma)$ cost that is too high to be used in an argument proving the convergence of the Born series. See the beginning of \S\ref{sec:3}. A sharper analysis of the joint $r\to 0^+,\sigma\to\infty$ limit is therefore required. 

Such an analysis has been supplied recently by Hintz \cite{Hintz2}, in great generality, though the basic ideas go back further --- see the literature review in \cite[\S1]{Hintz2}. In the exact Euclidean case, the homogeneity of the Laplacian suffices to reduce the problem to the limiting absorption principle, if one is willing to forgo the use of variable decay orders. Hintz's tools are thus overkill for this particular, constant order, case. However, we do use variable orders (for reasons mentioned below), so Hintz's tools prove useful even in the Euclidean case.

Blackboxing Hintz's estimates, \Cref{thm:main_0} will follow quickly.

The allowed long-range tail of the potential offers a further complication.  
For proving the convergence of the Born series in the presence of such
a potential, \emph{anisotropic} Sobolev spaces\footnote{Anisotropic Sobolev spaces are by now a standard
tool in microlocal analysis. Examples include the work of Hintz and Vasy on asymptotically hyperbolic and Kerr--de Sitter spaces
\cite{vasy2013asympthyp}\cite{hintz-vasy2015semilinear}, the work of Faure--Sj\"ostrand \cite{faure-sjostrand2011anosov} and Dyatlov--Zworski \cite{dyatlov-zworski2016dynamical} on Anosov flows, work of Vasy and Vasy--Molodyk on the Klein--Gordon equation \cite{VasyGrenoble}\cite{M-V-Feynman}, work of \cite{BaVaWu15} on the wave equation, and work of Hassell et al.\ on Schr\"odinger \cite{Gomes1, Gomes2}.}
turn out to be convenient.
To see the complication posed by long-range potentials, consider the Born approximation when $X=\overline{\bbR^3}$. The Born approximation is just the truncation of the Born series, \cref{eq:Born_series}, to the first two terms:
\begin{equation}
 R_V(\sigma \pm i0)(x,y) \approx \frac{e^{\pm i\sigma |x-y|}}{4 \pi|x-y|} - \int_{\bbR^3} \frac{e^{\pm i \sigma (|x-z|+|z-y|)} V(z)}{16 \pi^2 |x-z| |z-y|}  \dd^3 z .
\end{equation}
If $V(z)$ decays only like $1/z$ as $z\to\infty$, then the integral here is only \emph{conditionally} convergent, since the integrand is $\Omega (1/z^3)$ in this limit. 
So, getting any estimate at all requires taking into account the oscillations of the integrand. Microlocal tools do this in a natural way.

To round out our presentation, we give some quick applications:
\begin{itemize}
	\item  Perturbed plane waves and the S-matrix can both be expressed in terms of the resolvent, so the Born approximation for the latter yields the Born approximation for the former.
	We explain how this works for perturbed plane waves 
	in the asymptotically Euclidean case, in \S\ref{sec:multi-Coulomb}.
	\item Modulo a low energy obstruction, the convergence of the Dyson series for the (wave) Cauchy problem follows. Via the Vainberg method, this yields a result about resonances.  These results are discussed in \S\ref{sec:Dyson}.
\end{itemize}

The behavior and convergence of the Born series is a classic problem in the physics literature, with a long history: see, e.g.\ \cite{Dalitz}, \cite{JoPa51}, \cite{Ko54}, \cite{Davies}, \cite{We63}, \cite{Manning}, \cite{Bushell}, \cite{Gesztesy}.  Convergence for short-range central potentials with mild (e.g.\ Coulomb) singularities is long since established.  For instance Jost--Pais \cite[Appendix (III)]{JoPa51} show convergence of the Born series for perturbed plane waves at any finite frequency for a class of central potentials $V(r)$ going as $r^{-1}$ as $r \to 0$ and satisfying
\begin{equation}
    \int_0^\infty |rV(r)|\dd r < \infty.
    \label{eq:JostPais_cond}
\end{equation}
\cite[Eqs. 46, 47]{JoPa51}; this class thus includes screened Coulomb models such as the Yukawa and Hulth\'en potentials, but excludes long-range potentials such as the actual Coulomb potential or $V(r)=\langle r\rangle^{-2}$ (which fail \cref{eq:JostPais_cond}). 
Work of Kohn \cite{Ko54} and Gesztesy--Thaller \cite{Gesztesy} elucidates the convergence at high-energy of the partial wave expansion for short-range central potentials.
The convergence in $\Im \sigma>0$ was  extended by Weinberg in \cite[Section III]{We63} to what is described as ``almost any interesting two-body interaction;'' this work in particular showed that for $V \in L^2$, not necessarily central, the Born series for the resolvent converges; however the author leaves open the analytically more delicate question of what happens as $\Im \sigma \to 0$ (see the end of Section III).

Besides the fact that we work on general (non-trapping) asymptotically conic manifolds, our treatment here offers greater generality than the works cited above in the following ways:
\begin{itemize}
\item We consider long-range Coulomb-type potentials; this entails describing convergence in somewhat exotic spaces with variable Sobolev orders (``anisotropic'' Sobolev spaces).
\item We consider potentials that are not spherically symmetric, and indeed allow for multiple singularities; correspondingly, we make no use of partial wave expansions.
\item Detailed Sobolev spaces provide refined information on the convergence of the Born series in different operator-norm topologies.
\end{itemize}

\begin{remark*}
	We will exclusively talk about convergence in operator norm, but it is also interesting to \emph{fix} $f(x)$ (independent of $\sigma$) and study the convergence of $\mathsf{B}(\sigma \pm i0) f$ as $\sigma\to\infty$. 
	The convergence should be more rapid, with each successive term suppressed by $O_f(\sigma^{-2})$ relative to the previous, owing to semiclassical elliptic estimates (in which $h=1/\sigma$ is the ``semiclassical parameter'').

	By contrast, if $f=f(-;\sigma)$ is a family of functions which oscillate \emph{like the Liouville--Green ansatz} $e^{i\sigma \bullet }$ -- i.e., if $f$ has semiclassical wavefront set on the semiclassical characteristic set of $\triangle-\sigma^2$ -- then instead of elliptic estimates one is forced to use semiclassical propagation estimates. These lose a factor of $\sigma$ relative to the elliptic estimates. So, for such $f$, the successive terms in the Born series are $O_f(\sigma^{-1})$, saturating the uniform bounds \cref{eq:uniform}.\footnote{Allowing Coulomb singularities might possibly worsen these estimates by $O(\sigma^\varepsilon)$, for arbitrarily small $\varepsilon>0$. We do not rule this out, as is apparent from \cref{eq:uniform}.}
\end{remark*}

\section*{Acknowledgements}
The first author was supported by an NSF Postdoctoral Fellowship; the early stages of this work began with conversations with Andr{\'a}s Vasy and Peter Hintz. 
The second author acknowledges partial support from NSF grants DMS-2054424 and DMS-2452331 and from Simons Foundation grant MPS-TSM-00007464. 

We are grateful for the incisive comments of the anonymous referee.

\section{Definition of the Born series}

Recall that we let
$[X;\text{markings}]$
denote the new manifold-with-boundary obtained by blowing up $X$ at the marked points in the sense of \cite[Section 18]{Me94}; this process replaces each of these points with its spherical normal bundle (i.e., locally passes to spherical coordinates).  For any manifold-with-boundary $Y$, we let $S^0(Y)$ denote the space of conormal distributions of order $0$ on $Y$; letting $\mathcal{V}_{\mathrm{b}}(Y)$ denote the space of vector fields tangent to $\pa Y$, this space is described as follows:
\begin{equation}
f \in S^0(Y) \Longleftrightarrow V_1\dots V_N f \in L^\infty(Y), \text{ for all } N \in \mathbb{N},\ V_1,\dots, V_N \in \mathcal{V}_{\mathrm{b}}(Y).
\end{equation}
Specializing to $Y=[X;\text{markings}]$, the space consists of functions that enjoy Kohn--Nirenberg symbol estimates of order zero near infinity, and have iterated $L^\infty$ regularity under application $r \pa_r$ and $\pa_\theta$ in spherical coordinates near each of the marked points.

The purpose of this section is to show that provided the potential $V$ satisfies 
\begin{equation} 
V \in r^{-1} S^0([X;\text{markings}]),
\end{equation} 
the individual terms in \cref{eq:Born_series} make sense.
For this purpose, we use the ``sc-b'' Sobolev spaces
\begin{equation}
	H_{\mathrm{sc-b}}^{m,\mathsf{s},\ell}([X;\mathrm{markings}]),
\end{equation} 
which by abuse of notation we abbreviate to $H_{\mathrm{sc-b}}^{m,\mathsf{s},\ell}(X).$
 The space locally agrees with the sc-Sobolev space $H^{m,\mathsf{s}}_{\mathrm{sc}}$ \cite{Me94} at large $r$ and a b-Sobolev space \cite{Me93}  at small $r$.  Thus,
\begin{itemize}
    \item near $\partial X$, it is a generalization to non-Euclidean $X$ of
ordinary weighted Sobolev spaces $\langle r \rangle^{-\mathsf{s}} H^m(\bbR^n)$;
    \item at the various marked points, it is a b-Sobolev space, measuring
finite order conormal regularity  with respect to the point in question, with a (local near $0$) weight $r^\ell$ relative to $L^2$. Here $r$ is the distance to a marked point.
\end{itemize}
Thus $\mathsf{s}$ measures the decay as $r\to\infty$ and $\ell$ is the decay as $r\to 0$ (i.e., at marked points). As usual, $m$ is the number of derivatives controlled. By ``derivatives,'' what is meant is a vector field 
\begin{equation*} 
\in \langle r \rangle^{-1} \calV_{\mathrm{b}}([X;\text{markings}]),
\end{equation*} 
i.e.\ is what Melrose calls in \cite{Me94}   a sc-vector field near $\partial X$ and a b-vector field near the marked points. An sc-vector field is just a b-vector field with one extra order of decay; in the Euclidean case, the constant-coefficient vector fields are sc-.
For $u$ to lie in $H_{\mathrm{sc-b}}^{m,\mathsf{s},\ell}$ means that 
\begin{equation} 
    V_1\dots V_m u\in H_{\mathrm{sc-b}}^{0,\mathsf{s},\ell}
\end{equation}
for any such vector fields $V_1,\dots,V_m$.

A subtlety here is that the decay order $\mathsf{s} \in C^\infty({}^{\mathrm{sc}}\overline{T}^* X)$ is allowed to be a function \emph{on the sc-phase space} $\smash{{}^{\mathrm{sc}}\overline{T}^* X}$ \cite{Me94, MelroseGeometric}, the ball-bundle over $X$ whose fibers are parametrized by frequency. Microlocal analysis on the sc-phase space works almost exactly like ordinary microlocal analysis on the cosphere bundle of a closed manifold, using a pseudodifferential calculus $\Psi_{\mathrm{sc}}$ quantizing symbols on the sc-phase space. One has an associated notion of sc-wavefront set, \begin{equation*} 
	\operatorname{WF}_{\mathrm{sc}}^{\bullet}(u)\subset  \partial ({}^{\mathrm{sc}}\overline{T}^* X)
\end{equation*} 
for $u\in \mathcal{S}'(X)$, which measures the obstruction to $u$ lying in the function space labeled by `$\bullet$.' In particular, $\operatorname{WF}_{\mathrm{sc}}(u)$ is the obstruction to $u$ being Schwartz, and, for $m,s\in \bbR$, $\operatorname{WF}_{\mathrm{sc}}^{m,\mathsf{s}}(u)$ denotes the obstruction to $u$ lying in the Sobolev space $H_{\mathrm{sc}}^{m,\mathsf{s}}(X)$.  
This Sobolev space 
agrees with $H_{\mathrm{sc-b}}^{m,\mathsf{s},\ell}$ away from the marked points and is $H^m(X^\circ)$ near the marked points.

The variable order $\mathsf{s}$ is said to satisfy the incoming/outgoing Sommerfeld conditions (which depend on $\sigma$) if it is strictly monotonic under the geodesic flow\footnote{This means that $\sqrt{H_p \mathsf{s}}$ is a well-defined symbol, where $H_p$ is the generator of the geodesic flow. This is the Hamiltonian vector field associated with the symbol $p$. Here $p$ is the sc-principal symbol of the operator $P$, meaning the symbol modulo lower-order sc-symbols.} on the characteristic set of the operator $P$ and satisfies the two threshold conditions 
\begin{itemize}
	\item $\mathsf{s}>-1/2$ on the incoming radial set $\operatorname{WF}_{\mathrm{sc}}(e^{-i\sigma r})$, 
    \item $\mathsf{s}<-1/2$ on the outgoing radial set $\operatorname{WF}_{\mathrm{sc}}(e^{i\sigma r})$,
\end{itemize}
in addition to the technical requirement that $\mathsf{s}$ is constant in neighborhoods of the two radial sets.

\emph{Note that the Sommerfeld conditions depend on $\sigma$.}

\begin{remark*}
    Because this depends on $\sigma$, the $\mathsf{s}$ below will depend implicitly on $\sigma$. Later on, $\mathsf{s}$ can be taken to be a constant function of the ``semiclassical scattering frequency'' variables $h\xi=\xi/\sigma$. 
\end{remark*}

Membership in the Sobolev space $H_{\mathrm{sc}}^{m,\mathsf{s}}(X)$ therefore allows an outgoing spherical wave $r^{-(d-1)/2} e^{i\sigma r}$ but forbids an incoming spherical wave $r^{-(d-1)/2} e^{-i\sigma r}$.
The Sommerfeld conditions characterize in this way the monochromatic radiation emitted by a source $f$.
See \cite{BaVaWu15} for an introduction to variable order (a.k.a.\ anisotropic) Sobolev
spaces. We refer to \cite{Me94, MelroseGeometric} for further discussion of radial point conditions
in asymptotically Euclidean scattering theory.
For more discussion of the sc-b Sobolev space alongside other Sobolev spaces that we employ below,
see \S\ref{sec:Sobolevs}.

\begin{example}
	In the asymptotically Euclidean case $X=\bbR^d_x$, take $\mathsf{s}$ of the form $\mathsf{s}=-1/2- \varepsilon \psi(\theta \cdot \xi/\sigma)$ near $\partial X$, where $\varepsilon >0$ is small, $\theta\in \bbS^{d-1}$ is the spherical angle, $\xi$ is the momentum coordinate dual to $x$, and $\psi\in C^\infty(\bbR)$ is monotonic, identically $-1$ near $-1$, and identically $+1$ near $+1$. This satisfies all of the requirements, except for the technical condition (imposed so as to avoid appealing to the sharp G{\aa}rding inequality in the proof of positive commutator estimates) that $|H_p \mathsf{s}|$ be the square of a symbol. This extra condition is easy to satisfy by building $\psi$ from explicit bump functions; see the proof of positive commutator estimates in \cite[Chp.\ 8]{HintzBook}, and \cite[Eq. 8.43]{HintzBook} specifically.
\end{example}

We now recall that the limiting resolvent $R_V(\sigma\pm i0)=\lim_{\varepsilon \to 0^+} R_V(\sigma \pm i \varepsilon)$ exists as a map $C_{\mathrm{c}}^\infty(X^\circ)\to \calS'(X)$ (``limiting absorption principle'');  for long-range potentials without a singularity, this is in \cite{Me94}, \cite{VasyOverloaded}. The same proof applies in the setting with Coulomb singularities (and even dipole singularities) \cite{Hintz2}, as the proof of the next proposition (though written for $R_0$) shows. We thus take as given the existence of the limiting resolvent; what is at stake here is whether and in what manner the Born series converges to it. However, this paper can also be read as giving an alternative construction of the limiting resolvent as a map between suitable function spaces. One can then show that the object constructed in this manner is indeed the limit of the resolvent as the spectral parameter approaches the real axis.

We now address more refined estimates on the mapping properties of the free resolvent on a scattering manifold with marked points.
\begin{proposition}
	For each $\sigma >0$, the free resolvent $R_0(\sigma \pm i0)$  extends (uniquely) to a bounded linear map
	\begin{equation}
		R_0(\sigma \pm i0) : H_{\mathrm{sc-b}}^{m-2,\mathsf{s}+1,\ell-2}(X) \to H_{\mathrm{sc-b}}^{m,\mathsf{s},\ell}(X)
	\end{equation}
	whenever $m\in \bbN$, $\mathsf{s}$ satisfies the incoming/outgoing (depending on $\pm$) Sommerfeld conditions, and $\ell \in (2-d/2,d/2)$. 
\end{proposition}
\begin{proof}
	Let $\{\chi_j\}_{j=1}^\#\subset C^\infty(X)$ denote a partition of unity such that the support of each $\chi_j$ contains at most one marked point $p_j$. Then, 
	\begin{equation}
		R_0(\sigma\pm i0)  = \sum_{j=1}^\# M_{\chi_j}\circ 	R_0(\sigma\pm i0) \circ M_{\chi_j} + \sum_{j\neq k}  M_{\chi_k}\circ 	R_0(\sigma\pm i0) \circ M_{\chi_j}. 
		\label{eq:decomposition}
	\end{equation}
	We have split the sum into a ``diagonal'' part and an ``off-diagonal'' part. 
	\begin{itemize}
		\item To handle the diagonal part, we can use the case of a single marked point, where we claim that
		\begin{equation} 
			R_0(\sigma \pm i0) : H_{\mathrm{sc-b}}^{m-2,\mathsf{s}+1,\ell-2}([X;p_j]) \to H_{\mathrm{sc-b}}^{m,\mathsf{s},\ell}([X;p_j])
		\end{equation} 
		for all $m\in \bbR$, $\mathsf{s}$ satisfying the relevant Sommerfeld conditions, and $\ell \in (2-d/2,d/2)$ (implicit in \cite{Hintz2}). This amounts to an estimate of $u$ in terms of $\triangle_g u$ in the $\mathrm{sc}-\mathrm{b}$-spaces. These estimates follow straightforwardly from stitching together (via an elliptic estimate in the interior) the standard sc- estimates at large-$r$ \cite[Proposition 10]{Me94} and elliptic b-estimates at small-$r$ \cite[Theorem 5.60]{Me93}. The restriction $2-d/2<\ell<d/2$ on $\ell$ can be understood as arising from separation of variables and analysis of the resulting indicial roots.  Imposing the regular boundary condition at $r=0$ requires that $\ell>2-d/2$. The requirement $\ell< d/2$ is needed because general solutions of the PDE have no better decay as $r\to 0$.
		\item  For the off-diagonal part, use the same estimate, together with 
		\begin{equation}
			M_{\chi_k} : H_{\mathrm{sc-b}}^{m,\mathsf{s},\ell}([X;p_j]) \to H_{\mathrm{sc}}^{m,\mathsf{s}}(X)  . 
		\end{equation}
		The point is that away from $p_j$, being in $H_{\mathrm{sc-b}}^{m,\mathsf{s},\ell}([X;p_j])$ just means being in $H_{\mathrm{sc}}^{m,\mathsf{s}}(X)$.  
		Next, we claim that there is an inclusion (technically
                induced by the pullback under the blowdown map)
		\begin{equation} 
			H_{\mathrm{sc}}^{m,\mathsf{s}}(X)\subseteq  H_{\mathrm{sc-b}}^{m,\mathsf{s},0}([X;p_j]):
		\end{equation} 
		Away from the marked points, this is tautological, so the key point is comparing having $m$ orders of b-regularity at the marked points with $m$ orders of ordinary Sobolev regularity. It suffices to consider the case of the marked point being the origin in $\bbR^d$. 
		Then, $f\in H^m_{\mathrm{b}}([\bbR^d,\{0\}])$ nearby means that $Lf\in L^2$ nearby for $L$ formed out of up to $m$-fold products of 
		\begin{equation}
			r \partial_r = \sum_{j=1}^d x_j \partial_{x_j} , \quad \partial_{\theta_{j,k}}=x_j \partial_{x_k} - x_k \partial_{x_j}. 
		\end{equation}
		This follows from being in the ordinary Sobolev space $H^m$ locally; the extra factors of $x_\bullet$ only help.
		
		Thus,
		\begin{equation} 
			M_{\chi_k}\circ 	R_0(\sigma\pm i0) \circ M_{\chi_j} : H_{\mathrm{sc-b}}^{m-2,\mathsf{s}+1,\ell-2}(X) \to H_{\mathrm{sc-b}}^{m,\mathsf{s},0}(X).
		\end{equation}
		The last thing to do is improve this to $H_{\mathrm{sc-b}}^{m,\mathsf{s},\ell}(X)$. But this follows immediately from elliptic regularity -- indeed, $R_0(\sigma\pm i0) \circ M_{\chi_j}$ outputs things which are smooth near $p_k$, which means being in 
		\begin{equation} 
			r^{d/2-} H^m_{\mathrm{b}} \subset r^{\ell}H^m_{\mathrm{b}}
		\end{equation} 
		locally. 
	\end{itemize}
\end{proof}

Now we let $\calX_\pm$, $\calY_\pm$ denote the union of all of codomains resp.\ domains in the previous proposition:
\begin{definition}
    For each $\sigma>0$, let
    \begin{align}
        \calX_\pm &= \bigcup_{m\in \bbN} \bigcup_{\ell\in (2-d/2,d/2)} \bigcup_{\mathsf{s} \text{ $\pm$-Sommerfeld}} H_{\mathrm{sc-b}}^{m,\mathsf{s},\ell}(X), \\
        \calY_\pm &= \bigcup_{m\in \bbN} \bigcup_{\ell\in (2-d/2,d/2)} \bigcup_{\mathsf{s} \text{ $\pm$-Sommerfeld}} H_{\mathrm{sc-b}}^{m-2,\mathsf{s}+1,\ell-2}(X).
    \end{align}
    \end{definition}
Then $R_0(\sigma \pm i0) : \calY_\pm \to \calX_\pm$. (The extension of $R_0(\sigma\pm i0)$ to two different $\mathrm{sc}-\mathrm{b}$ Sobolev spaces, having the same marked points but different orders $m,\mathsf{s},\ell$, has to agree on their overlap, since we can find a bigger $\mathrm{sc}-\mathrm{b}$ Sobolev space containing both and satisfying the hypotheses of the previous proposition.)

Note that the multiplication operator $M_V$ satisfies $M_V: \calX_\pm
\to \calY_\pm$. Indeed, for each $m,\mathsf{s},\ell$ as above,
\begin{equation}
	M_V :  H_{\mathrm{sc-b}}^{m,\mathsf{s},\ell}(X)\to  H_{\mathrm{sc-b}}^{m,\mathsf{s}+1,\ell-1}(X) \subseteq  H_{\mathrm{sc-b}}^{m,\mathsf{s}+1,\ell-2}(X) \subseteq \calY_\pm . 
\end{equation}

We are thus in a position to define the Born series.
    \begin{definition}
	For each $\sigma>0$, the Born sequence $\{B_j(\sigma\pm i0)\}_{j=0}^\infty \subset \calL(\calY_\pm,\calX_\pm) $ is defined recursively by: 
	\begin{itemize}
		\item $B_0(\sigma\pm i0) = R_0(\sigma \pm i0)$,
		\item for $j\geq 0$,  assuming that we have already defined $B_j(\sigma \pm i0) :\calY_\pm \to \calX_\pm$, then 
		\begin{equation}
			B_{j+1}(\sigma\pm i0) = -R_0(\sigma \pm i0) M_V B_j(\sigma \pm i0),
			\label{eq:Bj_def} 
		\end{equation}
		which also lies in $\calL(\calY_\pm ,\calX_\pm)$. 
	\end{itemize}
\end{definition} 

Here, by $\calL(\calY_\pm ,\calX_\pm)$, we just mean the set of linear maps $\calY_\pm \to \calX_\pm$. 

\section{Hintz's semiclassical cone estimates}
\label{sec:3}
In order to estimate the terms in the Born series, we will need to control the limiting resolvent $R_0(\sigma \pm i0)$ on functions of the form $f(x)/r$ for $f$ in semiclassical Sobolev spaces. 
To begin, we recall Hardy's inequality:
\begin{equation}\label{eq:Hardy}
\Big\lVert \frac{f(x)}{r} \Big\rVert_{L^2} \lesssim \lVert f \rVert_{H^1}.
\end{equation}
Recall also that in the classic cutoff free resolvent estimate $\sigma R_0(\sigma \pm i0) : L^2_{\mathrm{c}}\to L_{\mathrm{loc}}^2$ (where the constant is uniform as $\sigma\to\infty$, see, e.g., \cite[Prop.\ 2.1]{Bu02}, \cite[Thm.\ 3.1]{DyZw19}), we can trade the $\sigma^{-1}$ for a derivative, by making use of a semiclassical elliptic estimate, or more simply by estimating $\lVert u \rVert_{H^1}$ in terms of $\langle u,P_0 u \rangle$ by integrating by parts. This yields that $R_0(\sigma\pm i0):L^2_{\mathrm{c}} \to H^1_\loc$ uniformly in $\sigma$.
Combining this with \cref{eq:Hardy} yields the uniform estimate
\begin{equation}
\Big\lVert \chi R_0(\sigma\pm i0)\frac{f(x)}{r} \Big\rVert_{H^1} \lesssim \lVert f \rVert_{H^1}.
\end{equation}
Note the lack of any decay as $\sigma\to\infty$.
Consequently, Hardy's inequality does not suffice to prove that successive terms in the Born series are suppressed as $\sigma\to\infty$. Something more sophisticated is required.

At this point, let us change notation slightly so as to match that used by Hintz; define 
\begin{equation}
h = 1/\sigma. 
\end{equation}
Then, the $\sigma\to \infty$ limit is the same as the $h\to 0^+$ limit. Hintz's estimates are phrased in terms of $h$, which is referred to as the ``semiclassical parameter.''

It turns out that, for this purpose, one wants to allow different
orders as $r\to 0^+$ for $h>0$ fixed and as $r,h\to
0^+$ together, proportionally. One natural scale of Sobolev spaces
accomplishing this is that of the \emph{semiclassical cone} Sobolev
spaces of Hintz \cite{Hintz1, Hintz2}. Hintz was only concerned with
the situation in compact subsets of $X^\circ$, so no decay order
$\mathsf{s}$ was present. Consequently, instead of Hintz's semiclassical cone
spaces, we work with a four-index family 
\begin{equation}
H_{\vee}(X) = \{H_{\vee}^{m,\mathsf{s},\ell,\alpha}(X) \}_{m,\ell,\alpha\in \bbR,\, \mathsf{s} \in C^\infty({}^{\mathrm{sc},\hbar}\overline{T}^* X ;\bbR ) }, 
\end{equation}which agree with Hintz's semiclassical cone spaces in compact subsets of $X^\circ$ and which agree with the ordinary semiclassical scattering Sobolev spaces $H_{\mathrm{sc},\hbar}^{m,\mathsf{s}}$ near $\partial X$. We note in particular the relationship to Hintz's semiclassical cone spaces for functions localized near marked points: fixing $K \subset X^\circ$ compact,
\begin{align} 
\begin{split} 
H_{\vee}^{m,*,\ell,\alpha}(K) &= H_{\mathrm{c},h}^{m,\ell,\alpha,0}(K) \\ 
h^b H_{\vee}^{m,*,\ell,\alpha-b}(K) &= H_{\mathrm{c},h}^{m,\ell,\alpha,b}(K)
\label{eq:Sob_abs}
\end{split} 
\end{align}
in the notation of \cite{Hintz2}.
(Again, see \S\ref{sec:Sobolevs} for the precise details.)
The $H_\vee$-spaces are indexed such that 
\begin{equation}
H_{\vee}^{0,0,0,0}(X) = L^2(X), 
\end{equation}
and increasing $m,\mathsf{s},\ell,\alpha$ results in more regular functions, i.e., smaller spaces.
The order $m$ is the number of derivatives controlled, which is mostly inessential here since we deal with the solution to an elliptic equation.  The other three orders $\mathsf{s},\ell,\alpha$, are of greater importance: these serve respectively as the number of orders of decay at spatial infinity, the number of orders of decay at the marked points (for $r\ll h$), and the number of orders of decay as $r,h\to 0$ together. Hintz's ch-spaces have one other index, the `$b$' in \cref{eq:Sob_abs}, standing for the number of orders of decay as $h\to 0$ for $r\gg 0$. As \cref{eq:Sob_abs} indicates, incrementing that index is the same as multiplying by an overall factor of $h$ while decrementing the $\alpha$ index. Consequently, that index is optional, and we choose not to use it, instead keeping explicit factors of $h$.

\begin{figure}
	\begin{tikzpicture}[scale=.85]
		\fill[gray!5] (-5,2.5) -- (0,2.5) -- (0,0)  arc(90:180:1.5) -- (-5,-1.5) -- cycle;
		\draw[dashed] (-5,2.5) -- (-5,-1.5);
		\node (zfp) at (.35,1) {$b$};
		\node (flr) at (-3,-1.7) {$\ell$};
		\node (mf) at (-.8,-.6) {$\alpha$};
		\node (sf) at (-2.5,2.7) {$s$};
		\draw (0,2.5) -- (0,0)  arc(90:180:1.5)  -- (-5,-1.5);
		\draw (-5,2.5) -- (0,2.5);
		\draw[->, darkred] (-.1,.1) -- (-.1,.75) node[left] {$r$};
		\draw[->, darkred] (-.1,2.4) -- (-.1,1.5) node[left] {$1/r$};
        \draw[->, darkred] (-.1,2.4) -- (-1.1,2.4) node[below left] {$h=1/\sigma$};
		\draw[->, darkred] (-.1,.1) to[out=180, in=30] (-.8,-.1) node[above left] {$\hat{r}^{-1}$};
		\draw[->, darkred] (-1.6,-1.4) -- (-2.4,-1.4) node[above] {$h$};
		\draw[->, darkred] (-1.6,-1.4) to[out=90, in=235] (-1.34,-.6) node[left] {$\hat{r}=r\sigma$};
	\end{tikzpicture}
	\caption{The manifold-with-corners $[X\times (0,\infty]_\sigma;\{r=0,\sigma=\infty\}]$ on which the polynomial weights defining the $\vee$-Sobolev spaces are defined. The various boundary hypersurfaces are associated with different Sobolev orders.}
	\label{fig}
\end{figure}
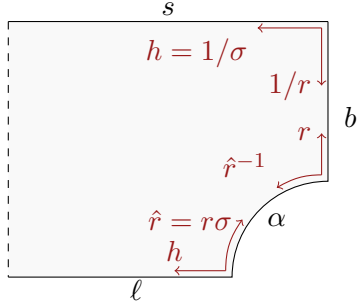

\begin{example}
	Take $m=0$ and $\mathsf{s}$ equal to some constant $s\in \bbR$. Then, the $\vee$-Sobolev space is 
	\begin{equation} 
	H_{\vee}^{0,s,\ell,\alpha}(X)(\sigma) = r^{\ell}(r+h)^{\alpha-\ell} \langle r \rangle^{-(s+\alpha)} L^2(X),
	\end{equation}
	where $h=1/\sigma$, where we have added `$(\sigma)$' on the left-hand side to emphasize that $H_{\vee}^{0,s,\ell,\alpha}(X)$ is a \emph{family} of $\sigma$-dependent function spaces, each of which is $\smash{H_{\mathrm{sc-b}}^{0,s,\ell}}$ at the level of topological vector spaces, but endowed with a $\sigma$-dependent norm. 
\end{example}
\noindent Thus, ignoring the subtlety that we will need the decay order $\mathsf{s}$ to be different in different regions of phase space (hence a \emph{variable} order), and ignoring the question of which weighted derivatives will be controlled, the $H_\vee$ spaces are just the ordinary $L^2$-based Sobolev spaces but allowing some simple weights, $r,\langle r\rangle, (r+h)$. 
The key weight that allows for the analysis of the joint $h,r\to 0^+$ limit is $r+h$. Because we also have $r$ as a weight, the $\vee$-Sobolev spaces can capture distinct decay rates as $r\to 0^+$ for $h>0$ fixed versus in the regime where $h\to 0^+$ too.

One has in this setting an analogous set of Sommerfeld criteria for
$\mathsf{s}$. Because only the index $\mathsf{s}$ is relevant to the
$r\to\infty$ regime, the definition in the setting with marked points is identical to that  without marked points.

\subsection{The basic estimate}

The next proposition bounds the \emph{free} resolvent with respect to the $\vee$-Sobolev spaces in the case where there is one marked point. Note that the metric -- and therefore the free Helmholtz operator itself -- has no singularity at the marked point. The point of this estimate is to control $u=R_0(\sigma \pm i0) f$ for $f$ which may have a singularity at the marked point, either for individual $\sigma<\infty$ or in the $\sigma\to\infty$ limit, and to understand the singularity or smoothness of $u$ there.
\begin{proposition}
	\label{prop:Hintz}
	Suppose that $X$ has only a single marked point. Then, for each $m\in\bbN$, $\ell \in (2-d/2,d/2)$, $\mathsf{s}$ satisfying the incoming/outgoing Sommerfeld conditions,\footnote{Note that $\mathsf{s}$ is a function of $\sigma$. So the proposition should be read as: fix $\mathsf{s}$, including its $\sigma$-dependence, once and for all; then, there exists a constant $C>0$ such that...} $\varepsilon>0$, and $\Sigma>0$, there exists a constant $C=C(X,m,\ell,\mathsf{s},\varepsilon,\Sigma)>0$ such that 
	\begin{equation}
	\lVert R_0(\sigma\pm i0) \rVert_{ H_{\vee}^{m-2,\mathsf{s}+1,\ell-2,-1/2}(X)\to   H_{\vee}^{m,\mathsf{s},\ell,1/2}(X) } \leq \frac{C}{\sigma^{1-\varepsilon}}
	\label{eq:Hintz}
	\end{equation}
	for all $\sigma\geq \Sigma$. 
\end{proposition}
\begin{proof}
        This is essentially a special case of \cite[Thm. 5.7]{Hintz2} (the only reason that it's not is that Hintz stated the case $X=\overline{\bbR^d}$); we present a variant of Hintz's proof. This involves sandwiching the propagation estimates in \cite[\S4]{Hintz2} with standard semiclassical scattering radial point estimates. Hintz's estimates are used to estimate the resolvent output away from the boundary at infinity $\partial X$ (and in particular they apply at and near the marked point), while scattering radial point estimates apply to the resolvent output near $\partial X$.

        The strategy of proof is as follows: first, we obtain an estimate at the incoming radial set, based only on background regularity/decay assumptions on the solution (``above threshold regularity'').  By standard methods, we can then propagate this estimate to the whole of phase space \emph{except} for those rays that emanate from the marked point, i.e., we have control on those rays that asymptote to the incoming radial set through $X\backslash\{p\}$ in backward time.  The essential technical novelty is thus to use an estimate of Hintz \cite{Hintz2} to propagate regularity through the marked point, whence we may continue to obtain an estimate on all of phase space in the customary manner.
	\begin{enumerate}[label=(\roman*)]
		\item Fix a microlocalizer $E \in \Psi_{\mathrm{sc},\hbar}^{0,0}$ such that the backwards geodesic flowout of the essential support 
        \begin{equation} 
            \operatorname{WF}'_{\mathrm{sc},\hbar}(E) \subset {}^{\mathrm{sc},\hbar}T^* X
        \end{equation}
        of $E$ does not pass over the sole marked point.  
        Let $\chi \in C^\infty(X)$ equal $1$ identically on the projection down to the base $X$ of the backwards flowout of the microsupport of $E$.  
        The standard\footnote{These estimates are folklore in the literature on radial points.
        In \cite{VasyGrenoble}, Vasy refers to them as the $\mathrm{sc,}\hbar$ analogue of radial point estimates and states that they hold, but they are not written out explicitly.
        The proof is standard; equation \eqref{eq:supposedly_standard} follows by taking apart the global commutator argument employed in the proof of the main theorem of \cite{VaZw00}.  In particular, we may take the commutator argument of Section 3 of that paper but instead of using the full symbol $q$ constructed in Section 4, we just use the part of it denoted $q_-$, microlocalized near the incoming set.  It follows from (4.16) of \cite{VaZw00} that the commutator with the quantization of $q_-$ produces top order terms of only one sign, and we have a radial point estimate near the incoming set, predicated on the function in question lying in an ``above-threshold'' Sobolev space. In \cite{VaZw00} this background regularity is guaranteed by the presence of an imaginary part to the spectral parameter (as the goal there is to prove the limiting absorption principle), while here we enforce it by including the \[h^{N}\lVert E_0 u\rVert_{H_{\mathrm{sc},\hbar}^{m,-1/2+\delta}}\] term on the RHS of the estimate.  Finiteness of this term guarantees the desired above-threshold regularity.  
        
        Having established the desired estimate for $E$ microsupported near the incoming radial set, the estimate can be extended by standard propagation arguments, as the Hamilton vector field is nonradial elsewhere. (Note that the outgoing radial set is all visible from the Coulomb singularity, i.e., is connected to the marked point by geodesics, so we are not trying to extend the estimate to the outgoing set at this stage.) In the notation of \cite{VaZw00}, we can accomplish this next step by repeating the commutator argument with commutants given by quantization of $q_\xi$ and $q_\partial$: this gives positive commutator estimates with error terms in the ``control region'' $\supp q_-$ from the previous estimate. A version of the microlocalized argument described here is well-exposited in \cite[Sec.~E.4.3]{DyZw19}, but is not in the context of the scattering calculus; the results described here can thus be derived from those by conic localization and Fourier transform (cf.\ \cite[Sec.~4]{Me94}).}
        propagation and radial point estimates in the setting without any cone/marked points yield
		\begin{equation}
			\lVert E u \rVert_{H_{\mathrm{sc},\hbar}^{m,\mathsf{s}} } \lesssim h \lVert \chi P u \rVert_{H_{\mathrm{sc},\hbar}^{m-2,\mathsf{s}+1} }  + h^{N}\lVert E_0 u\rVert_{H_{\mathrm{sc},\hbar }^{m,-1/2+\delta}}  + h^N \lVert u\rVert_{H_{\mathrm{sc},\hbar}^{-N,-N} } 
            \label{eq:supposedly_standard}
		\end{equation}
	   for any $\delta>0$, $N\in \bbR$, $E_0\in \Psi^{-\infty,0}_{\mathrm{sc},\hbar}(X)$ elliptic  on the incoming radial set.\footnote{\Cref{eq:supposedly_standard} also holds with $\chi$ replaced by an appropriate microlocalizer, but we will not make use of this.} 
       The estimate above holds for all $u\in \calS'(X)$, but the right-hand side may not be finite. (The same applies to the other formulas below.) For example, the right-hand side is infinite if $u$ does not have $-1/2+\delta$ orders of decay on the incoming radial set. Hence, this estimate is only useful when $u$ is ``above threshold'' in terms of its amount of incoming decay.
        
        For the purpose at hand, it will suffice to take $E$ elliptic on a small neighborhood of an annular region around the backwards flowout of the conormal/cosphere bundle of the marked point.
       See \Cref{fig:est_part1}.
\begin{figure}[h!]
    \centering
   \begin{tikzpicture}
    \filldraw[fill=lightgray!20] (0,0) circle (2);
        \begin{scope}[decoration={
				markings,
				mark=at position 0.62 with {\arrow[scale=1.5,>=latex, color=orange]{<}}}]
            \clip  (0,0) circle (2);
            \fill[lightgray!50] (0,0) circle (2);
            \filldraw[darkred, fill opacity=.5, pattern = north west lines, pattern color =darkred] (0,0) circle (1.4); 
            \filldraw[darkred, fill=lightgray!50] (0,0) circle (1);
            \fill[lightgray!20] (0,0) circle (.75);
            \draw[postaction={decorate}, black] (0,0) -- (0,2);
            \draw[postaction={decorate}, black] (0,0) -- (1.41,1.41);
            \draw[postaction={decorate}, black] (0,0) -- (0,-2);
            \draw[postaction={decorate}, black] (0,0) -- (2,0);
            \draw[postaction={decorate}, black] (0,0) -- (-2,0);
            \draw[postaction={decorate}, black] (0,0) -- (-1.41,1.41);
            \draw[postaction={decorate}, black] (0,0) -- (1.41,-1.41);
            \draw[postaction={decorate}, black] (0,0) -- (-1.41,-1.41);
            \clip (0,0) circle (1.4);
            \draw[postaction={decorate}, orange] (0,0) -- (0,2);
            \draw[postaction={decorate}, orange] (0,0) -- (1.41,1.41);
            \draw[postaction={decorate}, orange] (0,0) -- (0,-2);
            \draw[postaction={decorate}, orange] (0,0) -- (2,0);
            \draw[postaction={decorate}, orange] (0,0) -- (-2,0);
            \draw[postaction={decorate}, orange] (0,0) -- (-1.41,1.41);
            \draw[postaction={decorate}, orange] (0,0) -- (1.41,-1.41);
            \draw[postaction={decorate}, orange] (0,0) -- (-1.41,-1.41);
            \clip (0,0) circle (1);
            \draw[postaction={decorate}, black] (0,0) -- (0,2);
            \draw[postaction={decorate}, black] (0,0) -- (1.41,1.41);
            \draw[postaction={decorate}, black] (0,0) -- (0,-2);
            \draw[postaction={decorate}, black] (0,0) -- (2,0);
            \draw[postaction={decorate}, black] (0,0) -- (-2,0);
            \draw[postaction={decorate}, black] (0,0) -- (-1.41,1.41);
            \draw[postaction={decorate}, black] (0,0) -- (1.41,-1.41);
            \draw[postaction={decorate}, black] (0,0) -- (-1.41,-1.41);
        \end{scope}\fill (0,0) circle (1.5pt);
        \draw (0,0) circle (2);
        \node () at (0,2.3) {$X$};
        \node () at (-2,2.3) {(a)};
        \node[darkred] () at (-1.55,-.6) {$\calE$}; 
    \end{tikzpicture}
    \hspace{1em}
    \begin{tikzpicture}
    \filldraw[fill=lightgray!20] (0,0) circle (2);
        \begin{scope}[decoration={
				markings,
				mark=at position .23 with {\arrow[scale=1.5,>=latex, color=black]{>}}}]
            \clip  (0,0) circle (2);
            \fill[lightgray!50] (0,0) circle (2);
            \filldraw[darkred, fill opacity=.5, pattern = north west lines, pattern color =darkred] (0,0) circle (1.4); 
            \filldraw[darkred, fill=lightgray!50] (0,0) circle (1);
            \fill[lightgray!20] (0,0) circle (.75);
            \fill (0,0) circle (1.5pt);
            \draw[postaction={decorate}] (-2,0) -- (2,0);
        \draw[postaction={decorate}] (-2,0) to[out=40, in=140] (2,0);\draw[postaction={decorate}] (-2,0) to[out=-40, in=-140] (2,0);\draw[postaction={decorate}] (-2,0) to[out=-20, in=-160] (2,0);\draw[postaction={decorate}] (-2,0) to[out=20, in=160] (2,0);
        \draw[postaction={decorate}] (-2,0) to[out=70, in=110] (2,0);
        \draw[postaction={decorate}] (-2,0) to[out=-70, in=-110] (2,0);
        \end{scope}
        \begin{scope}[decoration={
				markings,
				mark=at position .23 with {\arrow[scale=1.5,>=latex, color=orange]{>}}}]
            \clip (0,0) circle (1.4);
            \clip (-2,-2) rectangle (0,2);
            \draw[postaction={decorate}, orange] (-2,0) -- (2,0);
        \draw[postaction={decorate}, orange] (-2,0) to[out=-20, in=-160] (2,0);
        \draw[postaction={decorate}, orange] (-2,0) to[out=20, in=160] (2,0);
        \clip (0,0) circle (1);
        \draw[postaction={decorate}] (-2,0) -- (2,0);
        \draw[postaction={decorate}] (-2,0) to[out=-20, in=-160] (2,0);
        \draw[postaction={decorate}] (-2,0) to[out=20, in=160] (2,0);
        \end{scope}
        \begin{scope}[decoration={
				markings,
				mark=at position .8 with {\arrow[scale=1.5,>=latex, color=black]{>}}}]
            \clip  (0,0) circle (2);
            \clip (0,-2) rectangle (2,2);
            \draw[postaction={decorate}] (-2,0) -- (2,0);
        \draw[postaction={decorate}] (-2,0) to[out=40, in=140] (2,0);\draw[postaction={decorate}] (-2,0) to[out=-40, in=-140] (2,0);\draw[postaction={decorate}] (-2,0) to[out=-20, in=-160] (2,0);\draw[postaction={decorate}] (-2,0) to[out=20, in=160] (2,0);
        \draw[postaction={decorate}] (-2,0) to[out=70, in=110] (2,0);
        \draw[postaction={decorate}] (-2,0) to[out=-70, in=-110] (2,0);
        \end{scope}
        \draw (0,0) circle (2);
        \node[gray] () at (0,-1.7) {$\operatorname{supp} \chi$};
        \node () at (-2,2.3) {(b)};
    \end{tikzpicture}
    \caption{The projection down to the base, $\calE$, of the essential support of $E$ (in {\color{darkred}red}), and $\supp \chi$. The marked point is in the middle ($\bullet$). Depicted in (a), (b)  are different sets of geodesics; (a) depicts the geodesics in the backwards flowout (flow\emph{in}) of the conormal bundle of the marked point. The essential support of $E$ is contained near the portions ({\color{orange}orange}) over $\calE$ of the lifts of these geodesics. In (b) are shown the incoming geodesics emanating from the left. Note that (lifted) geodesics which stay far from the marked point never pass through the essential support of $E$, even if they pass over $\calE$, because they never stray close enough to the flowin of the marked point.}
    \label{fig:est_part1}
\end{figure}
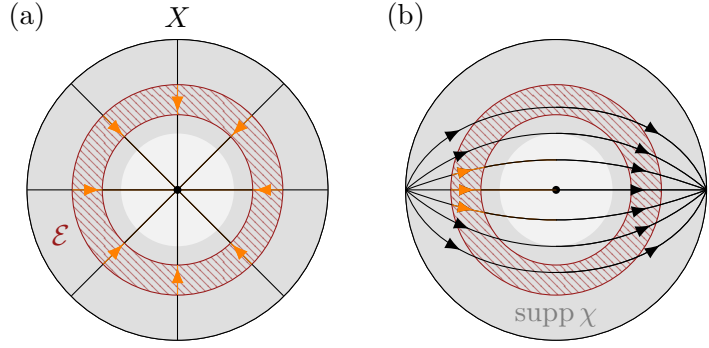
		\item  
		For $E$ as above, with Schwartz kernel compactly supported in $(X^\circ \backslash\{p\})^2$, and such that every geodesic which is incoming to the marked point passes through the elliptic set of $E$, 
		\begin{equation}
			\lVert \chi u \rVert_{H_\vee^{m,*,\ell,\frac{1+\varepsilon}{2}}} \lesssim h^{1-\varepsilon} \lVert  \tilde{\chi} P u \rVert_{H_\vee^{m-2,*,\ell-2,-\frac{1+\varepsilon}{2}} } + \lVert E u \rVert_{H_\hbar^{-M} } + h^\delta \lVert \tilde{\chi} u \rVert_{H_\vee^{-N,*,\ell,\frac{1+\varepsilon}{2}}},
            \label{eq:our_own_Hintz}
		\end{equation}
		for any $M,N\in \bbR$, $\varepsilon>0$, some $\delta>0$ depending on $\varepsilon$, and any
                $\chi,\tilde{\chi}\in C_{\mathrm{c}}^\infty(X^\circ)$ such
                that $\chi$ is supported in a small neighborhood of
                the marked point and $\tilde{\chi}=1$ identically on a
                geodesically convex set containing the support of $\chi$, the
                marked point, and the projection down to the base of the essential support of $E$.
                The conclusion and proof are essentially those of \cite[Prop.\ 5.2]{Hintz2}; the hypotheses are slightly weaker, but the Sobolev spaces appearing in the final estimate are the same, except we now have an extra 
                \begin{equation}\label{eq:new_error}
                     h^\delta \lVert \tilde{\chi} u \rVert_{H_\vee^{-N,*,\ell,\frac{1+\varepsilon}{2}}}
                \end{equation}
                term on the right-hand side to make up for the weaker hypotheses.\footnote{The main \emph{notational} difference between \cref{eq:our_own_Hintz} and \cite[Prop.\ 5.2]{Hintz2} is that what Hintz calls $P_{h,z}$ is what we call $h^2 P = \sigma^{-2} P$. One of these factors of $h$ cancels with the $h^{-1}$ on the right-hand side of the estimates in \cite[Prop.\ 5.2]{Hintz2}, and the other $h$ is the source of the $h^{1-\varepsilon}$ in front of $\lVert \tilde{\chi} Pu \rVert$ in \cref{eq:our_own_Hintz}. 
                Moreover, \cite[Thm.\ 5.7]{Hintz2} is stated in terms of $\calP_{h,1} = P_{h,1}+ i Q$, with $\lVert \calP_{h,1}u \rVert$ on the  right-hand side. Here the $Q$ term is a \emph{complex absorbing potential} which is elliptic on the whole characteristic set of $P_{h,1}$ in some region; moving this to the other side of the equation and using the triangle inequality gives an estimate of the form \eqref{eq:our_own_Hintz}, but with $Q$ replacing $E$. Since our estimate in \eqref{eq:our_own_Hintz} replaces this term with a microlocalized term $E$, however, we are in principle forced to revisit the proof of \cite[Thm.\ 5.7]{Hintz2}, which nevertheless goes through essentially verbatim to yield the result stated here. The main difference is that the additional semi-global error term \eqref{eq:new_error} cannot be absorbed in the LHS since the LHS of the estimate is now local.
                } In order to prove \cref{eq:our_own_Hintz}, we go back to \cite[Thm.\ 4.10.(1)]{Hintz2}.\footnote{That theorem is for $X$ a compact manifold-without-boundary, but this does not matter, because if $\chi_0\in C_{\mathrm{c}}^\infty(X^\circ)$ is identically $1$ on a large enough set, then $\chi \chi_0 u = \chi u$, $\tilde{\chi} Pu=\tilde{\chi} P (\chi_0 u)$, $Eu=E(\chi_0 u)$, and $\tilde{\chi} u = \tilde{\chi}(\chi_0 u)$. So, we can apply \cite[Thm.\ 4.10]{Hintz2} to $\chi_0 u$ to get \cref{eq:our_own_Hintz}. This is why we have left the sc-decay order unspecified in that equation.}
                The range of $\ell$ allowed in that theorem is given in \cite[Lem. 5.1]{Hintz2}; it is what is stated in the proposition statement here, $\ell\in (2-d/2,d/2)$. Also, our theorem here is stated without variable semiclassical order, so is lossy compared to \cite[Thm. 4.10]{Hintz2}. The deduction of the constant-order statement from the variable-order statement is exactly the same as in the proof of \cite[Prop.\ 5.2]{Hintz2}.
		\item  For any  
        \begin{itemize}
            \item $\chi,\chi_0 \in C^\infty(X)$ such that $\chi_0=1$ identically on the support of $\chi$, and
            \item $Q \in \Psi_{\mathrm{sc},\hbar}^{0,0}$ whose Schwartz kernel is supported away from the marked point (on projection to either factor) and such that every geodesic over $\operatorname{supp}\chi$ eventually passes through
                  the elliptic set of $Q$ while remaining where $\chi_0$ does not vanish,  
        \end{itemize}
        the estimate
		\begin{equation}
		\lVert \chi u \rVert_{H_{\mathrm{sc},\hbar}^{m,\mathsf{s}}} \lesssim h \lVert \chi_0 Pu \rVert_{H_{\mathrm{sc},\hbar}^{m-2,\mathsf{s}+1} } + \lVert Q u \rVert_{H_{\mathrm{sc},\hbar}^{m,\mathsf{s}} }+ h^N \lVert u\rVert_{H_{\mathrm{sc},\hbar}^{-N,-N} } 
		\end{equation}
        holds. 
		This involves the standard below-threshold radial point estimate in the setting without any cone/marked points.\footnote{Once again we refer the reader to \cite{VaZw00} for details: this estimate employs the commutant obtained from quantizing $q_+$ in the notation of that paper, with error terms now estimated by the output of the estimate in our previous step.}
	\end{enumerate}
        Owing to the non-trapping hypothesis, every backwards geodesic asymptotes either to the incoming radial set or the marked point, where we have control according to (i) and (ii). Consequently, control can be propagated throughout, including the outgoing radial set, according to (iii). Stitching together the estimates above yields
	\begin{equation}
	\lVert u \rVert_{H_\vee^{m,\mathsf{s},\ell,\frac{1+\varepsilon}{2}}} \lesssim h^{1-\varepsilon} \lVert P u \rVert_{H_\vee^{m-2,\mathsf{s}+1,\ell-2,-\frac{1+\varepsilon}{2}} } + h^N \lVert E_0 u \rVert_{H_{\mathrm{sc},\hbar}^{m,-1/2+\delta} } + h^{\delta} \lVert u \rVert_{H_\vee^{m,\mathsf{s},\ell,\frac{1+\varepsilon}{2}}}.
	\end{equation}
	The last term can be absorbed into the left-hand side, as can the penultimate term under the assumption that it is finite. 
Hence we obtain
	\begin{equation}
	 \lVert E_0 u \rVert_{H_{\mathrm{sc}}^{m,-1/2+\delta} }<\infty  \Longrightarrow 
	\lVert u \rVert_{H_\vee^{m,\mathsf{s},\ell,\frac{1+\varepsilon}{2}}} \lesssim h^{1-\varepsilon}  \lVert P u \rVert_{H_\vee^{m-2,\mathsf{s}+1,\ell-2,-\frac{1+\varepsilon}{2}} }.\label{foo.12.5}
	\end{equation}
	For any $u\in \calS'(X)$, if $\sigma$ is large
        enough and such that
        $\operatorname{WF}_{\mathrm{sc}}^{*,-1/2+\varepsilon}( u )$ is disjoint
        from the incoming radial set \emph{for that value of
        $\sigma$}, then we may choose $E_0$ as in \cref{eq:supposedly_standard} so that 
          \begin{equation} 
          \lVert E_0 u \rVert_{H_{\mathrm{sc}}^{m,-1/2+\delta} }<\infty,
            \end{equation} 
        hence
        the inequality on the right-hand side of \cref{foo.12.5} 
        holds. The key property of the limiting resolvent $R_0(\sigma
        \pm i0)$ is that its output $R_0(\sigma \pm i0)f$ has this
        property whenever $f \in \calY_\pm$ \cite{Me94}.\footnote{This is also true when there is a $1/\langle r\rangle$ potential present, as long as it's real-valued. See \cite[\S16]{Me93}\cite{VasyOverloaded}.} We thus conclude 
	\begin{equation}
	\lVert R_0(\sigma \pm i0) f \rVert_{H_\vee^{m,\mathsf{s},\ell,\frac{1+\varepsilon}{2}}} \lesssim h^{1-\varepsilon}  \lVert f \rVert_{H_\vee^{m-2,\mathsf{s}+1,\ell-2,-\frac{1+\varepsilon}{2}} }.
	\end{equation}
	We can drop the $\varepsilon$'s in the Sobolev orders, as doing this weakens the estimate.
\end{proof}

An especially important case of \Cref{prop:Hintz} is when $X$ is exact Euclidean space, with only a single marked point, in which case $R_0$ is the free Euclidean resolvent. Then a direct proof is straightforward if $\mathsf{s}$ is replaced by a constant $s_-<-1/2$ on the left-hand side and $s_+>-1/2$ on the right-hand side. One way of doing this is by using the homogeneity of the Euclidean resolvent. Rescaling $r$ to $\hat{r}=r\sigma$ reduces the energy $\sigma^2$-resolvent to the energy-one resolvent.

\begin{remark*}
	Hintz's argument works when the metric has a conic singularity at the marked point (as clearly stated before \cite[Thm. 5.7]{Hintz2}), or if a dipole potential is present (this being part of \cite[Thm. 5.7]{Hintz2}).  The argument breaks if more than one instance of these is present. The reason is that in order to control the resolvent output satisfactorily, one needs a priori control on all geodesics incoming towards the cone points/point dipoles. With only one, all incoming geodesics originate at spatial infinity, so the required control is provided by the limiting absorption principle. With \emph{multiple} cone points/point dipoles, 
	geodesics can originate at one marked point and terminate at another. 
	This is the unstable trapping explored e.g.\ in \cite{BaWu13}, \cite{HiWu20}.
\end{remark*}

\subsection{The estimate for multiple marked points}
 For any $p\in X^\circ$, we can apply the preceding proposition with $p$ being the marked point (being careful to note that the $\vee$-Sobolev spaces involved depend on $p$). This yields the following:
\begin{theorem}
	\label{thm:main_estimate}
	For each $m\in\bbN$, $\ell \in (2-d/2,d/2)$, $\mathsf{s}$ satisfying the incoming/outgoing Sommerfeld conditions, $\varepsilon>0$, and $\Sigma>0$, there exists a constant 
    \begin{equation} 
        C=C(X,m,\ell,\mathsf{s},\varepsilon,\Sigma)>0
    \end{equation}
    such that 
	\begin{equation}
	\lVert R_0(\sigma\pm i0) \rVert_{ H_{\vee}^{m-2,\mathsf{s}+1,\ell-2,-1/2}\to   H_{\vee}^{m,\mathsf{s},\ell,1/2} } \leq \frac{C}{\sigma^{1-\varepsilon}}
	\label{eq:misc_000}
	\end{equation}
	for all $\sigma\geq \Sigma$. 
\end{theorem}
\begin{proof}
	Let $\{\chi_j\}_{j=0}^\#\subset C^\infty(X)$ denote a partition of unity such that 
	\begin{itemize}
		\item the support of $\chi_0$ is disjoint from every marked point, 
		\item for $j\geq 1$, the support of $\chi_j$ contains at most one marked point and is disjoint from $\partial X$. 
	\end{itemize}
	The supports are as depicted in \Cref{fig:supports}. 
	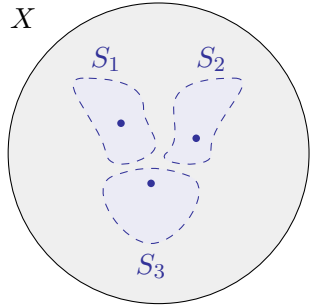
\begin{figure}[h!]
		\begin{tikzpicture}
			\begin{scope}
				\clip  (0,0) circle (2);
				\fill[lightgray!25] (-2,-2) rectangle (2,2);
				\filldraw[fill=darkblue!10, draw=darkblue, dashed] (-1,1) to[out=180, in=100] (-.7,0) to[in=220, out=-80] (-.1,0) to[out=40, in=-90] (-.2,.6) to[out=90,in=0] cycle;
				\filldraw[fill=darkblue!10, draw=darkblue, dashed] (1,1) to[out=0, in=100] (.7,0) to[in=220, out=-80] (.1,0) to[out=40, in=-90] (.2,.6) to[out=90,in=180] cycle;
				\filldraw[fill=darkblue!10, draw=darkblue, dashed] (-.1,-1.2) to[out=0, in=-20] (.4,-.3) to[out=160, in=20] (-.6,-.3) to[out=200,in=180] cycle;
				\node[darkblue] () at (-.7,1.25) {$S_1$};
				\node[darkblue] () at (.7,1.25) {$S_2$};
				\node[darkblue] () at (-.1,-1.5) {$S_3$};
			\end{scope}
			\fill[darkblue] (.5,.2) circle (.05);
			\fill[darkblue] (-.5,.4) circle (.05);
			\fill[darkblue] (-.1,-.4) circle (.05);
			\draw (0,0) circle (2);
			\node () at (-1.8,1.8) {$X$};
		\end{tikzpicture}
		\caption{The sets $S_j=\operatorname{supp}\chi_j$ in the proof of \Cref{thm:main_estimate} are chosen to only contain at most one marked point $p_j$ (blue) each. (They can even be chosen to be disjoint.) In this case, there are $\#=3$ marked points. The set $S_0$ containing $\partial X$ is not shown.}
		\label{fig:supports}
	\end{figure}
	Then, $R_0(\sigma \pm i0) = \mathtt{D} + \mathtt{O} + \mathtt{R} $
    for 
	\begin{equation}
			\mathtt{D} = \sum_{j=0}^\# M_{\chi_j}\circ 	R_0(\sigma\pm i0) \circ M_{\chi_j}, \qquad 
			\mathtt{O} = \sum_{\substack{j\neq k\\ j,k\neq 0}}  M_{\chi_k}\circ 	R_0(\sigma\pm i0) \circ M_{\chi_j} 
	\end{equation}
	and
	\begin{equation}
		\mathtt{R} = \sum_{j=1}^\# ( M_{\chi_0} \circ R_0(\sigma \pm i0) \circ M_{\chi_j}+ M_{\chi_j} \circ R_0(\sigma \pm i0) \circ M_{\chi_0} ). 
	\end{equation}
	The estimate for the ``diagonal part,'' $\mathtt{D}$  of $R_0(\sigma \pm i0)$, is immediate from \Cref{prop:Hintz}, as is the estimate for $\mathtt{R}$, which involves all of the other terms whose Schwartz kernels may have support at $\partial (X\times X)$. It is only the ``off-diagonal'' terms 
	$M_{\chi_j}\circ 	R_0(\sigma\pm i0) \circ M_{\chi_k} \in
        \mathtt{O}$, for nonzero $j\neq k$, that need to be controlled.
        These contain the information about how the various marked points
        communicate. 
	
	Let $U = M_{\chi_j}\circ 	R_0(\sigma\pm i0) \circ M_{\chi_k}$. This satisfies the PDE $P_0\circ U = F$ for 
	\begin{equation}
		F = M_{\chi_j \chi_k} + [P_0,M_{\chi_j}] \circ  R_0(\sigma \pm i0) \circ M_{\chi_k}.  
	\end{equation} 
	It follows from \Cref{prop:Hintz}, applied using $p_k$ as the marked point, that the operator 
	\begin{equation}
		[P_0,M_{\chi_j}]\circ  R_0(\sigma\pm i0) \circ M_{\chi_k}  = \underbrace{[\triangle,M_{\chi_j}]}_{\in \sigma \operatorname{Diff}^1_\hbar } \circ R_0(\sigma\pm i0) \circ M_{\chi_k}  
	\end{equation}
	satisfies 
	\begin{equation} 
		\lVert [P_0,M_{\chi_j}]\circ  R_0(\sigma\pm i0) \circ M_{\chi_k}  \rVert_{ H_{\vee}^{m-2,*,\ell-2,-1/2}\to   H_{\vee}^{m-1,*,*,*} } \lesssim \sigma^\varepsilon.
	\end{equation} 
	So, the ``forcing'' $F$ satisfies 
	\begin{equation}
		\lVert F \rVert_{ H_{\vee}^{m-2,*,\ell-2,-1/2}\to   H_{\vee}^{m-2,*,*,*} } \lesssim \sigma^\varepsilon.
        \label{eq:misc_kkk}
	\end{equation}

	The support of the output of $F$ does not intersect $\partial
        X$, so the fact that $P_0\circ U=F$ implies $U=R_0(\sigma \pm i0)\circ F$. Nor does the support of the output of $F$ contain any marked point. 
	So, noting that the output of $U= M_{\chi_j}\circ R_0 \circ M_{\chi_k}$ is supported away from every other marked point besides $p_j$, we can apply \Cref{prop:Hintz} 
	with $p_j$ as the marked point to conclude, from \cref{eq:misc_kkk}, that 
	\begin{equation}
		\lVert U \rVert_{ H_{\vee}^{m-2,*,\ell-2,-1/2}\to   H_{\vee}^{m,*,\ell,1/2} } \lesssim \frac{1}{\sigma^{1-2\varepsilon}}.
	\end{equation}
	Renaming $\varepsilon \to \varepsilon/2$, we are done.
\end{proof}

\begin{remark*}
	The crux of the argument above, controlling $M_{\chi_j}\circ R_0(\sigma\pm i0) \circ M_{\chi_k}$, is based on the following  algebraic manipulation:
	\begin{align*}
		M_{\chi_j} R_0 M_{\chi_k} = R_0 P_0 M_{\chi_j} R_0 M_{\chi_k}&= R_0 (  M_{\chi_j}P_0 R_0 M_{\chi_k} + [P_0, M_{\chi_j}]R_0 M_{\chi_k}  )   \\
		&= R_0 (  M_{\chi_j}M_{\chi_k} + [P_0, M_{\chi_j}]R_0 M_{\chi_k}  ) = R_0 F,
	\end{align*}
	where for brevity we abbreviated $R_0=R_0(\sigma \pm i0)$. 
	
	Now $F$ was controlled via \Cref{prop:Hintz} with $p_k$ as the marked point, whereas $R_0$, when applied to $F$, was controlled via the same proposition with $p_j$ as the marked point instead. Key therefore was our freedom to choose the marked point with respect to which the proposition is stated. 
	It is therefore essential that we are working with the \emph{free} Helmholtz operator on $X$.
	If the coefficients of $P_0$ had a singularity somewhere, for instance if the metric had a cone point, or if we tried to group some singular part of the potential with $P_0$ 
	then we could not apply \Cref{prop:Hintz} with any marked point \emph{other than} the cone point. This is the essential point where our setup differs from that in \cite{BaWu13}.
\end{remark*}

\section{Convergence of the Born series}
Using \Cref{thm:main_estimate}, it is easy to deduce the convergence of the Born series. First:
\begin{proposition}
	\label{prop:q_sep}
     For each $m\in\bbN$, $\mathsf{s}$ satisfying the incoming/outgoing Sommerfeld conditions, $\ell\in (2-d/2,d/2)$, $\varepsilon>0$, and $j_0\in \bbN$,  
	there exists a constant $C>0$ such that, whenever $j\geq j_0$, 
	\begin{equation}
	\lVert B_j(\sigma \pm i0) \rVert_{H_{\vee}^{m-2,\mathsf{s}+1,\ell-2,-1/2}\to H_{\vee}^{m+2j_0,\mathsf{s},\ell,1/2}(X) } \leq \Big( \frac{C}{\sigma^{1-\varepsilon}} \Big)^{j+1}
	\end{equation}
	for $\sigma\geq 1$. 
\end{proposition}
\begin{proof}
	The definition, \cref{eq:Bj_def}, of $B_j$ was $B_j(\sigma \pm i0) = R_0(\sigma \pm i0) (-M_V R_0(\sigma \pm i0))^{j}$.
	Combining 
	\begin{itemize}
		\item $\lVert M_V \rVert_{H_{\vee}^{m,\mathsf{s},\ell,1/2}\to H_{\vee}^{m,\mathsf{s}+1,\ell-1,-1/2}  } \lesssim 1$
        with the conclusion
		\item  $\lVert R_0(\sigma\pm i0) \rVert_{ H_{\vee}^{m-2,\mathsf{s}+1,\ell-2,-1/2}\to   H_{\vee}^{m,\mathsf{s},\ell,1/2} } \lesssim \sigma^{-(1-\varepsilon)}$
		of \Cref{thm:main_estimate},
	\end{itemize}
	we have $\sigma^{1-\varepsilon} M_V R_0(\sigma \pm i0) :
        H_{\vee}^{m-2,\mathsf{s}+1,\ell-2,-1/2} \to
        H_{\vee}^{m,\mathsf{s}+1,\ell-1,-1/2} $, uniformly. 
     We only need that 
     \begin{equation*}
	     \sigma^{1-\varepsilon} M_V R_0(\sigma \pm i0) :
	     H_{\vee}^{m-2,\mathsf{s}+1,\ell-2,-1/2} \to
	     H_{\vee}^{m,\mathsf{s}+1,\ell-2,-1/2} 
     \end{equation*}
     uniformly.    
     In other words, there exists some constant $C_1>0$ such that 
     \begin{equation*} 
     	\lVert M_V R_0(\sigma \pm i0) \rVert_{H_{\vee}^{m-2,\mathsf{s}+1,\ell-2,-1/2} \to
     	H_{\vee}^{m,\mathsf{s}+1,\ell-2,-1/2}} \leq \frac{C_1}{\sigma^{1-\varepsilon}} .
    \end{equation*}     
    Combining this with the uniform bound on $R_0$, it follows that 
    \begin{equation*}
    	\lVert B_{j_0}(\sigma \pm i0) \rVert_{H_{\vee}^{m-2,\mathsf{s}+1,\ell-2,-1/2}\to H_{\vee}^{m+2j_0,\mathsf{s},\ell,1/2}(X) } \leq C_0 \sigma^{-(1-\varepsilon)(j_0+1) }, 
    \end{equation*}
    for some constant $C_0>0$. Since $B_j(\sigma\pm i0) = B_{j_0}(\sigma \pm i0) (-M_V R_0(\sigma \pm i0))^{j-j_0}$, 
    \begin{multline*}
    \lVert B_j(\sigma \pm i0) \rVert_{H_{\vee}^{m-2,\mathsf{s}+1,\ell-2,-1/2}\to H_{\vee}^{m+2j_0,\mathsf{s},\ell,1/2}(X) } \leq \lVert B_{j_0}(\sigma \pm i0) \rVert_{H_{\vee}^{m-2,\mathsf{s}+1,\ell-2,-1/2}\to H_{\vee}^{m+2j_0,\mathsf{s},\ell,1/2}(X) } \\ \times  \lVert M_V R_0(\sigma \pm i0) \rVert_{H_{\vee}^{m-2,\mathsf{s}+1,\ell-2,-1/2}\to H_{\vee}^{m,\mathsf{s}+1,\ell-2,-1/2}(X) }^{j-j_0} \leq C_0 C_1^{j-j_0} \sigma^{-(1-\varepsilon)(j+1)}.  
    \end{multline*}
    Choosing $C>0$ large enough such that $C^{j+1}\geq C_0 C_1^{j-j_0}$ for all $j\geq j_0$, we have arrived at the claimed mapping property of $B_j$.
\end{proof}

\begin{theorem}[Convergence of Born series]
	\label{thm:main}
    For $m\in\bbN$, $\ell \in (2-d/2,d/2)$, and $\mathsf{s}$ as above, there exists a $\Sigma>0$ such that if $\sigma>\Sigma$, then 
		the Born series $\sum_{j=0}^\infty B_j(\sigma \pm i0)$
                converges in 
		\begin{equation}
		\calL(H_{\mathrm{sc-b}}^{m-2,\mathsf{s}+1,\ell-2 }(X), H_{\mathrm{sc-b}}^{m,\mathsf{s},\ell}(X) )
		\end{equation}
		(in the operator-norm topology) to the resolvent $R_V(\sigma \pm i0)$, with, for any $\varepsilon>0$, the bound 
        \begin{equation}
            \lVert R_V(\sigma \pm i0) \rVert_{H_{\vee}^{m-2,\mathsf{s}+1,\ell-2,-1/2 }(X)\to  H_\vee^{m,\mathsf{s},\ell,1/2}(X) } = O \Big(\frac{1}{\sigma^{1-\varepsilon}}\Big)  
        \end{equation}
        as $\sigma\to\infty$.

		Moreover, the tail of the Born series is increasingly regularizing: for any $\epsilon,\delta>0$,  for $\sigma>\Sigma$, for $\Sigma>1$ large enough, the tail of the series converges to zero in the norm topology on 
        \begin{equation} 
        \calL( H_{\mathrm{sc-b}}^{m-2,1/2+\epsilon,\ell-2 }(X), H^{M,-1/2-\epsilon,d/2-\delta}_{\mathrm{sc-b}}(X) ),
        \end{equation} 
        for any $M\in \bbR$. Moreover, for any $K\in \bbN$, if $N\gg 1$ is large enough, then the tail
        \begin{equation}
            \sum_{j=N}^\infty B_j(\sigma \pm i0)  
            \label{eq:Born_tail}
        \end{equation}
        is, as $\sigma\to\infty$, $O(1/\sigma^K)$ as an operator $H_{\mathrm{sc-b}}^{m-2,1/2+\epsilon,\ell-2 }(X)\to H^{M,-1/2-\epsilon,d/2-\delta}_{\mathrm{sc-b}}(X) )$. 
\end{theorem}
\begin{proof}By \Cref{prop:q_sep}, the Born series and its tail are geometrically convergent in the stated topologies if $\sigma$ is large enough. In fact, it converges in
\begin{equation}
		\calL(H_{\vee}^{m-2,\mathsf{s}+1,\ell-2, -1/2 }(X), H_\vee^{m,\mathsf{s},\ell, 1/2}(X) ),  
	\end{equation} 
	to an $O(1/\sigma^{1-\varepsilon})$ operator.
	Applying $P_V=\triangle+V-\sigma^2$ to the Born series, we can do so term-by-term, so whatever the Born series converges to, it produces a right-inverse to $P_V$. Because the limiting resolvent is uniquely characterized by producing a solution in $H^{*,\mathsf{s},*}_{\mathrm{sc}}$, the result of summing the Born series must be the actual resolvent. 

    When proving that the tails of the Born series are regularizing, one has to convert $\vee$-regularity (which involves iterated regularity under the application of vector fields with an $h$ out front) to ordinary $\mathrm{sc-b}$-regularity (which is $h$-independent); $\vee$-regularity is weaker than $\mathrm{sc-b}$-regularity. However, any finite number of orders of $\vee$-regularity can be upgraded to the same number of orders of $\mathrm{sc-b}$-regularity \emph{at a cost of some power $\sigma^k$ of $\sigma$} -- see \eqref{eq:sobolevinclusion1}, \eqref{eq:sobolevinclusion2} below. Fortunately, if we go far enough out in the Born series -- that is, if $N\gg 1$ in \cref{eq:Born_tail} is large enough -- then all of the terms in the Born series are $O(1/\sigma^k)$ in operator norm between the relevant $\vee$-Sobolev spaces.
\end{proof}

\begin{remark*}
	The exact same argument shows that if  $V\in \langle r \rangle^{-1} S^0([X;\text{markings}]) +  r^{-1}L^\infty_{\mathrm{c}}(X^\circ)$, then the Born series converges in 
	\begin{equation} 
	\calL(H_{\mathrm{sc-b}}^{0,1/2+\varepsilon,\ell-2 }(X), H_{\mathrm{sc-b}}^{0,-1/2-\varepsilon,\ell}(X) ) = \calL(\langle r \rangle^{-\ell+3/2-\varepsilon} r^{\ell-2} L^2(X), \langle r \rangle^{-\ell+1/2+\varepsilon} r^{\ell} L^2(X) )
	\end{equation} 
	to the actual resolvent, with a corresponding uniform bound in the $H_\vee$-spaces. The point is that we can handle arbitrary $L^\infty_{\mathrm{c}}(X^\circ)$ potentials, at the cost of giving up on the semiclassical smoothing of $B_j$.
\end{remark*}

\section{Application: series formula for Coulomb plane waves}
\label{sec:multi-Coulomb}

The purpose of this section is to illustrate how the Born series, in the form presented above, can be deployed to understand perturbed plane waves.\footnote{A \emph{perturbed plane wave} is a solution to $Pu=0$ that, for $r\gg 1$, looks like a plane wave plus an outgoing scattered wave. Physicists use perturbed plane waves to model what happens when a wave deflects off of a potential. A perturbed plane wave is the correct model (as opposed to, say, a perturbed spherical wave) when the incoming wave is much wider than whatever the wave is deflecting off of. Perturbed plane waves can also be used to construct the S-matrix \cite{ZwMe}.

What this means precisely depends on whether the potential is short-range or Coulomb. If the potential is short-range, then to top order it is still the case that $u(\bfx)\approx e^{i \sigma x}$, where the difference lies in the anisotropic sc-Sobolev space describing outgoing waves; in other words, the scattered wave may be viewed as a perturbation when viewed far away from the potential poles.  (For simplicity, we only consider plane waves moving in the $x$-direction.)

By contrast, in the presence of a (long-range!) Coulomb potential, we cannot find $u\approx e^{i\sigma x}$ such that $Pu= 0$. Instead, $u\approx e^{i\sigma x} (r-x)^{-i\mathsf{Z}/2\sigma}$ is the relevant replacement (with $\approx$ having the same meaning as in the previous paragraph), and the scattered wave cannot be so readily teased apart from the plane wave. This has to do with the logarithmic divergence of Coulomb trajectories from free trajectories.} 
Rather than aim for maximal generality, we will focus on a single example of historical interest, a superposition of multiple Coulomb potentials. We restrict attention to Euclidean space: $X=\overline{\bbR^3}$ with $\#\in \bbN^+$. The Green function for the free Helmholtz equation (with the outgoing radiation condition) is then
\begin{equation}
    R_0(\sigma +i0)(\bfx,\bfx') = \frac{e^{i\sigma |\bfx-\bfx'| } }{4\pi |\bfx-\bfx'|}.
\end{equation}
Expressions involving $R_0$ below are therefore explicit integrals.

Consider $\#\in \bbN^+$ marked points
$\bfx_1,\dots,\bfx_\# \in \bbR^3$. Now consider the
``Born--Oppenheimer'' operator 
\begin{equation}
P = - \frac{\partial^2}{\partial x^2}-\frac{\partial^2}{\partial y^2}-\frac{\partial^2}{\partial z^2} \underbrace{- \sum_{n=1}^\# \frac{\mathsf{Z}_n}{|\bfx-\bfx_n|}}_V - \sigma^2,  
\label{eq:P_Coulomb_section}
\end{equation}
where $\mathsf{Z}_n\in \bbR\backslash \{0\}$ and $\bfx=(x,y,z)$. A solution $u$ to $Pu=0$ describes the wavefunction of a (nonrelativistic) electron of energy $\sigma^2$ moving under the influence of $\#$ point charges $\mathsf{Z}_n$ clamped in place, or approximated as such. 

\begin{remark*}
    According to our sign conventions, a positive $\mathsf{Z}_\bullet$ means an attractive force, and a negative $\mathsf{Z}_\bullet$ means a repulsive force. However, nothing we say will depend on the sign of $\mathsf{Z}_\bullet$. 
\end{remark*}

\begin{example}[Exact Coulomb plane waves \cite{Temple}]
	Consider the case with only a single Coulomb potential, i.e.\ $\#=1$. Then, without loss of generality, we can take $p=p_1$ to be the origin, and we abbreviate $\mathsf{Z}_1=\mathsf{Z}$. Consider 
	\begin{equation}
	u(\bfx) = e^{i\sigma x} v(r-x), \quad 2sv''(s) + 2(1-is\sigma ) v'(s)+ \mathsf{Z} v(s) =0,
	\end{equation} 
	where $r=\lvert \bfx \rvert$. Then $u$ satisfies $Pu=0$.
	The ODE that $v$ satisfies is a form of the confluent hypergeometric equation \cite[\href{http://dlmf.nist.gov/13.2.E1}{Eq.\ 13.2.1}]{NIST}, with $a,b$ parameters $a=i\mathsf{Z}/2\sigma$ and $b=1$. 
	\emph{Which} confluent hypergeometric function $v$ should be (there is a two-dimensional space of possibilities) is determined by enforcing smoothness of $u$ near $r=x$, away from the origin (since $u$ solves an elliptic equation with smooth coefficients).
	This forces $v$ to be proportional to \emph{Kummer's confluent hypergeometric function} $M(a,b,i s \sigma)$ \cite[\href{http://dlmf.nist.gov/13.2.E2}{Eq.\ 13.2.2}]{NIST}, also known as ${}_1F_1$.
        (Indeed, the usual choice of second linearly independent\footnote{$M(a,b,z)$ is independent from Tricomi's confluent hypergeometric function $U(a,b,z)$ unless $a$ is a nonpositive integer, in which case $U(a,b,z)$ manages to be analytic at $z=0$. For us, $a\in i \bbR$, so we do not encounter the exceptional $a$. We are interested in the $a\to 0$ limit, however.} solution to the confluent hypergeometric equation, $U(a,b,is \sigma )$, has a logarithmic singularity as $s=0$ \cite[\href{http://dlmf.nist.gov/13.2.E19}{Eq.\ 13.2.19}]{NIST}.)
    The constant of proportionality is just a normalization constant. We thus take
	\begin{align}
	\begin{split} 
		v(s) &=  C M\Big( \frac{i\mathsf{Z}}{2\sigma},1,is \sigma \Big), \qquad C= (-i\sigma)^{\frac{i\mathsf{Z}}{2\sigma}} \Gamma\Big(1 - \frac{i\mathsf{Z}}{2\sigma} \Big),\\
		u(\bfx) &= Ce^{i\sigma x}M\Big( \frac{i\mathsf{Z}}{2\sigma},1,i(r-x) \sigma \Big) \in C^\infty(\bbR^3\backslash \{0\}).\end{split}
		\label{eq:Coulomb_exact}
	\end{align} 
    \begin{remark*} Because $M(0,-,z)=1$, 
    when $\mathsf{Z}=0$ we recover the ordinary plane wave, $u= e^{i\sigma x}$. 
    \end{remark*}
	It follows from the large-argument expansions of $M$ \cite[\href{http://dlmf.nist.gov/13.7.E2}{Eq.\ 13.7.2}]{NIST} that 
	\begin{equation}
	u(\bfx) = e^{i\sigma x} \underbrace{\exp \Big[- \frac{i\mathsf{Z}}{2\sigma}\log(r-x) \Big]}_{=(r-x)^{-i\mathsf{Z}/2\sigma}} + O\Big(\frac{1}{r} \Big)  
	\label{eq:u}
	\end{equation}
	as $r\to\infty$ within any backwards cone. See \Cref{fig:Mplot}. Evidently, the approximation $u(\bfx)\approx e^{i\sigma x}$ misses a logarithmically oscillating term $(r-x)^{-i\mathsf{Z}/2\sigma}$.

    \begin{figure}
        \centering
        \includegraphics[width=0.7\linewidth]{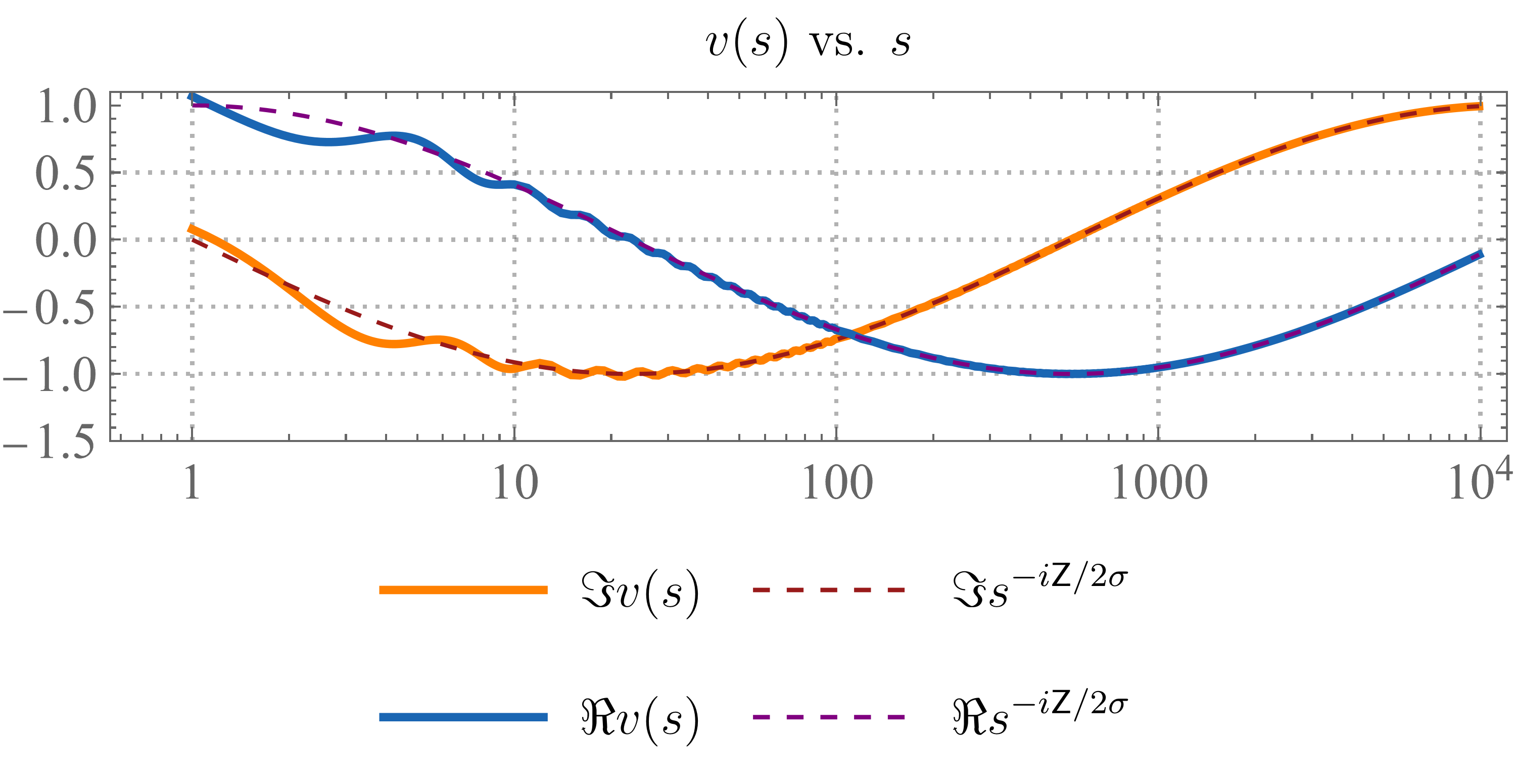}
        \caption{The real and imaginary parts of the function $v(s)$, plotted on a log-linear plot vs. $s$. Also shown are the real and imaginary parts of $s^{-i\mathsf{Z}/2\sigma}$, which $v(s)$ approximates at large $s$.}
        \label{fig:Mplot}
    \end{figure}
\end{example}

As $r\to\infty$, the potential $V(\bfx)$ can be expanded in powers of $1/r$:
\begin{equation}
-V(\bfx)=\sum_{n=1}^\# \frac{\mathsf{Z}_n}{| \bfx-\bfx_n|} = \frac{\mathsf{Z}}{r} + \frac{\bfx \cdot \bfa}{r^3} + O\Big(\frac{1}{r^3}\Big),
\label{eq:multipole}
\end{equation}
where 
\begin{itemize}
	\item $\mathsf{Z}=\sum_{n=1}^\# \mathsf{Z}_n \in \bbR$
	denotes the total charge (a.k.a.\ monopole moment), which can be positive, negative, or zero, and 
	\item $\bfa \in \bbR^3$ is the dipole moment. Concretely, $\bfa = \sum_{n=1}^\# \mathsf{Z}_n \bfx_n $.
\end{itemize}
This is the multipole expansion.
Below, we will use 
\begin{equation*}
W = -V - \frac{\mathsf{Z}}{r} -\frac{\bfx\cdot \bfa}{r^3}
\end{equation*}
to denote the $O(1/r^3)$ error term in \cref{eq:multipole}.

If the total charge $\mathsf{Z}$ is nonzero, then the dipole moment $\bfa$ depends on the choice of origin: 
fix $\bfx_0\in \bbR^3$, and consider moving each point charge $\mathsf{Z}_j$ from $\bfx_j$ to $\bfx_j+\bfx_0$. Then, $\bfa$ changes by $\mathsf{Z} \bfx_0$. We can use this freedom to reduce to the case where the dipole moment is zero: $\bfa=0$. 

The plan for the rest of this section is as follows:
\begin{itemize}
    \item As a warm-up, we will consider in \S\ref{sub:short} the case where the 
monopole and dipole moments both vanish: 
\begin{equation} 
    \mathsf{Z},\bfa=0.
\end{equation}
In this case, the potential is short-range: \begin{equation*} 
V=O\Big(\frac{1}{r^3}\Big).
\end{equation*} 
    \item  We will consider in \S\ref{sub:long} the case where $\mathsf{Z}\neq 0$ and $\bfa=0$. Then, the potential is Coulomb (with quadrupole and higher-order corrections): 
    \begin{equation*} 
    	V=-\frac{\mathsf{Z}}{r}+O\Big(\frac{1}{r^3}\Big).
    \end{equation*}  
    \item  The remaining case, where $\mathsf{Z}=0$ but $\bfa\neq 0$, requires special attention and will be dispatched in \S\ref{sub:dipole}. This is the case of a dipole potential $V\sim 1/r^2$ (with higher-order corrections).
\end{itemize}
Note that this trichotomy is only about the $r\to\infty$ behavior. All of our potentials are at worst Coulombic as $r\to 0$.

\subsection{The short-range case}
\label{sub:short}
Suppose that $\mathsf{Z},\bfa=0$. 
Then, the potential consists of at worst a quadrupole potential, decaying like $1/r^3$ as $r\to\infty$, plus faster decaying terms:
\begin{equation}
-V(\bfx) =  \frac{\bfx^\intercal Q \bfx}{r^5} + O\Big(\frac{1}{r^4} \Big) ,\quad Q\in \mathbb{R}^{3\times 3}
\label{eq:quadropole}
\end{equation}
as $r\to\infty$, where $Q$ is a symmetric 3-by-3 matrix, the ``quadrupole moment'' of the charge distribution. The potential $V$ is therefore short range.

Let $u_0 = e^{i \sigma x}$. Then, $Pu_0 \neq 0$ (unless all the charges are zero), but $Pu_0\approx 0$, in the sense that the quantity 
\begin{equation} 
    f_0\overset{\mathrm{def}}{=} Pu_0 = e^{i\sigma x}  V(\bfx) = e^{i\sigma x} O\Big(\frac{1}{r^3} \Big)  
\end{equation} 
is decaying fast enough to lie in the domain of the resolvent. Specifically:
\begin{itemize}
    \item as $r\to\infty$, we have $O(1/r^3)$ decay, which means lying in $\langle r \rangle^{-3/2+\varepsilon} L^2\{r\gg 1\}$, and 
    \item as $\bfx\to \bfx_n$, we have $O(1/r)$ growth, which is not severe enough to prevent lying in $L^2$ locally.
\end{itemize}
Thus, $f_0 \in \langle r \rangle^{-3/2+\varepsilon} L^2(\bbR^3)= H_{\mathrm{sc}}^{0,3/2-}(X)$. This is contained in the domain of the resolvent, $\calY_+$.
We can therefore define an \emph{exact} solution to $Pu=0$ by writing 
\begin{equation} 
u = u_0  - R_{V}(\sigma + i0) f_0 \in e^{i\sigma x} + \calX_+ .
        \label{eq:misc_a0}
\end{equation}
This is the perturbed plane wave.

The Born series for the resolvent then yields a series representation for $u$: 
\begin{proposition}
	\label{prop:short_range_main}
    There exists a $\Sigma>0$ such that, if $\sigma>\Sigma$, then 
    \begin{equation}
        u = u_0 - R_0(\sigma +i0)\sum_{j=0}^\infty  (-M_V R_0(\sigma +i0))^jf_0 
    \end{equation}
    on every compact subset away from the marked points. If $J\gg 1$ is sufficiently large, each term in the tail with $j \geq J$ is $O(1/\sigma^{1-\varepsilon})$ relative to the previous term.
\end{proposition}
\begin{proof}
    This follows from plugging \Cref{thm:main} into \cref{eq:misc_a0} and using Sobolev embedding. We do not even need uniform bounds on $f_0$.
\end{proof}

Under the assumption $\bfa=0$, the norm $\lVert f_0 \rVert_{ \langle r \rangle^{-1+\varepsilon}L^2} $ is bounded uniformly in $\sigma$, for any $\varepsilon>0$.  Consequently, the previous proposition gives bounds on $u$ which are uniform in the $\sigma\to\infty$ limit. 
 
\begin{remark*}
    As the arguments above make clear, we only made use of (i) the short-range nature of the potential as $r\to\infty$ and (ii) the Coulomb singularity of the potentials as $r\to 0$. Consequently, \Cref{prop:short_range_main} applies, \textit{mutatis mutandis}, to other potentials with the same large- and small-$r$ behavior. For example, it applies to any superposition of Yukawa potentials $V(\bfx)= e^{-\gamma r}/r$, $\gamma>0$.
\end{remark*}

\subsection{The long-range case}
\label{sub:long}
If either of the monopole moment $\mathsf{Z}$ or the dipole moment $\bfa$ is nonzero, then the function $f_0 = e^{i\sigma x} V(\bfx)$ defined in the last subsection is
\begin{equation}
    f_0 = \begin{cases}
        \Omega(1/r) & (\mathsf{Z}\neq 0), \\ 
        \Omega(1/r^2) & (\mathsf{Z}=0,\; \bfa\neq 0),
    \end{cases}
    \label{eq:f00}
\end{equation}
as $r\to\infty$ 
(excepting a measure-zero set of directions in the latter case). Consequently,  
\begin{equation} 
f_0\notin 
\begin{cases}
L^2(\bbR^3) & (\mathsf{Z}\neq 0), \\ 
\langle r \rangle^{-1/2} L^2(\bbR^3)   &(\mathsf{Z}=0,\,\bfa\neq 0).
\end{cases}
\end{equation} 
For $f_0$ to lie in the domain of the resolvent, $\calY_+$, we would want $f_0\in \langle r \rangle^{-1/2-\varepsilon} L^2$ (for some $\varepsilon>0$).  When $\mathsf{Z}=0$, our $f_0$ has borderline decay -- it just barely fails to lie in $\calY_+$. In the $\mathsf{Z}\neq 0$ case, it is not close. Thus we cannot use \cref{eq:misc_a0} as a useful first ansatz to obtain the perturbed plane wave $u$ by applying the Born series, since we cannot even apply the operators to this function.

The problem is that $u_0=e^{i\sigma x}$ is just not a very good approximation to the true perturbed plane wave. 
It is not hard to find a better approximation, as we now explain.  
Let $\mathsf{Z}\neq 0$. (We will discuss in the next subsection what to do when $\mathsf{Z}=0$ and $\bfa\neq 0$.)

A natural guess, that should yield a better approximation when  $r\gg 1$, is to take the exact Coulomb plane wave with charge $\mathsf{Z}$. 
Intuitively, one expects that, far away, a collection of point charges is indistinguishable from a single point charge whose charge is the total charge $\mathsf{Z}$. This is what the multipole expansion \cref{eq:multipole} makes precise, and extends.
Thus instead of defining $u_0=e^{i\sigma x}$, we let
\begin{equation}
u_0(\bfx) = C \chi \Big( \frac{1}{r} \Big)  e^{i\sigma x}M\Big( \frac{i\mathsf{Z}}{2\sigma},1,i(r-x) \sigma \Big),
\label{eq:Coulomb_ansatz}
\end{equation}
where $\chi \in C_{\mathrm{c}}^\infty(\bbR)$ is identically one near the origin. (Thus $\chi(1/r)$ localizes near infinity.) Then $Pu_0=f_0$ for
\begin{equation}
f_0 = -  W(\bfx) u_0 + \underbrace{[P, \chi(r^{-1})] \overbrace{(u_0 / \chi(r^{-1}))}^{\mathclap{Ce^{i\sigma x} M(i\mathsf{Z}/2\sigma,1,i(r-x)\sigma)}} }_{\in C_{\mathrm{c}}^\infty(\bbR^3)}, 
\label{eq:f0}
\end{equation}
where $W$ is the $O(1/r^3)$ error term in \cref{eq:multipole}. The absence of a $O(1/r^2)$ term in \cref{eq:f0} is due to $\bfa=0$, the case to which we reduced.

In \Cref{lem:Kummer_bound} we prove that the $M(a,b,-)$ factor in $u_0$ lies in $L^\infty(\bbR^3)$; see also the $z\to\infty$ asymptotics of $M(a,b,z)$ in \cite[\href{http://dlmf.nist.gov/13.7.E2}{Eq. 13.7.2}]{NIST}. Consequently, $f_0$ is $O(1/r^{3})$ as $r\to\infty$, i.e., it has two whole orders more decay as $r\to\infty$ than \cref{eq:f00} yields in the general long-range case.
It follows that

\begin{equation} 
f_0\in \langle r \rangle^{-3/2+\varepsilon}L^2\subset \calY_+.
\end{equation}
(Note that we have not needed to say anything about the sc-wavefront set of $f_0$.)
We can then define our perturbed plane wave $u$ by \cref{eq:misc_a0}. 
\begin{proposition}
    The conclusion of \Cref{prop:short_range_main} applies, except with $u_0$ defined by \cref{eq:Coulomb_ansatz} and $f_0 = Pu_0$. 
    \label{thm:Coulomb_main}
\end{proposition}
\begin{proof}
    By
    \Cref{thm:main}.  
\end{proof}

Note that the bound in \Cref{lem:Kummer_bound} is uniform as $\sigma\to\infty$. An easier argument using the formula 
\begin{equation}
    \partial_z M(a,1,z) = \frac{1}{\Gamma(1-a)\Gamma(a)} \int_0^1 e^{zt} t^a (1-t)^{-a} \dd t
\end{equation}
(which holds when $-1<\Re a < 1$) gives a uniform bound on the $C_{\mathrm{c}}^\infty$ term in \cref{eq:f0} where a derivative falls on the $M$-function. There is also a $C_{\mathrm{c}}^\infty$ term where a derivative falls on $e^{i\sigma x}$, producing an $O(\sigma)$. So, the previous proposition yields an $O(\sigma)$ bound on the full perturbed plane wave, away from the marked points and spatial infinity.

    

\subsection{The dipole case}
\label{sub:dipole}
We now turn to the case of vanishing total charge $\mathsf{Z}$, but nonvanishing dipole moment $\bfa$.
Because $\mathsf{Z}=0$, the ansatz $u_0$ above becomes simply $u_0=e^{i\sigma x}$.

We saw above that for such potentials, this ansatz, while a better approximation to the perturbed plane wave than in the charged case, still fails to solve $Pu_0\approx 0$ to a sufficiently good degree of approximation for $f_0=Pu_0$ to lie in the known domain of the resolvent, $\calY_+$. The problem is that $f_0$ still has slightly too little decay; namely, $f_0=O(1/r^2)$, which is borderline for lying in $\calY_+$ in three dimensions.

The situation is actually slightly better than just portrayed, microlocally speaking, because to lie in the domain $\calY_+$ requires only having an above-threshold amount of decay \emph{on the incoming radial set} $\calR= \operatorname{WF}_{\mathrm{sc}}(e^{-i \sigma \langle r \rangle})$: 
\begin{equation} 
\operatorname{WF}_{\mathrm{sc}}^{*,1/2+\varepsilon}(f_0)\cap \calR=\varnothing
\end{equation}
for some $\varepsilon>0$. The point is that $f_0$ satisfies this, except over the backwards direction.
Indeed, since $P$ is a sc-differential operator, and sc-differential (or pseudodifferential) operators do not spread sc-wavefront sets, 
\begin{equation}
    \operatorname{WF}_{\mathrm{sc}}(f_0)\subset \operatorname{WF}_{\mathrm{sc}}(u_0)=\operatorname{WF}_{\mathrm{sc}}(e^{i\sigma x}),
\end{equation}
and $\operatorname{WF}_{\mathrm{sc}}(e^{i\sigma x})$ intersects $\calR$ only over the backwards direction, where the incoming spherical wave $e^{-i\sigma r}$ agrees with the plane wave $e^{i\sigma x}$.

To fix the problem in the backwards direction, let 
\begin{equation}
u_1(x,y,z) = \underbrace{e^{i\sigma x}}_{=u_0}- \chi \Big(\frac{1}{r}\Big) \chi \Big( \frac{\sqrt{y^2+z^2}}{x}\Big) 1_{x\leq 0} \frac{e^{i\sigma x}}{2i \sigma} \underbrace{\int_{-\infty}^x \frac{s a_1+y a_2+z a_3}{(s^2+y^2+z^2)^{3/2}}  \dd s}_{-v(\bfx)} , 
\label{eq:modified_ansatz}
\end{equation}
where $\chi \in C_{\mathrm{c}}^\infty(\bbR)$ is as above.
(The integral is absolutely convergent on the support of the prefactor. This does not include the origin.) This ``corrects'' $u_0$ by adding to it a term 
\begin{equation}
	e^{i\sigma x} \phi(\bfx),\quad \phi(\bfx) = \chi \Big(\frac{1}{r}\Big) \chi \Big( \frac{\sqrt{y^2+z^2}}{x}\Big) 1_{x\leq 0} \frac{1}{2i \sigma} v(\bfx) 
\end{equation}
supported near the backwards direction.

Our improved ansatz $u_1$ will not satisfy $Pu_1=0$ exactly, but it will satisfy $Pu_1=f_1$ for some $f_1$ whose sc-wavefront set lies off the incoming radial set except in exactly the backwards direction. As we will check below, $f_1$ is $O(1/r^3)$ near the backwards direction, so it has one more order of decay than $f_0$, which was $O(1/r^2)$.
Thus, $f_1$ is decaying enough in the backwards direction to lie in the domain of the resolvent, $\calY_+$. We can therefore define our perturbed plane wave $u$ by
\begin{equation}
    u = u_1 - R_{V}(\sigma +i0) f_1.
\end{equation}
This satisfies $Pu=0$, so we have produced a solution to the PDE.

\begin{remark*}[Motivating the ansatz]
	The ansatz $u_1=e^{i\sigma x} + e^{i\sigma x} (2i\sigma)^{-1}v$ comes from the heuristic that when applying $\triangle-\sigma^2$ to $e^{i\sigma x} v$, for symbolic $v \in S^0$, the term where two derivatives fall on the exponential $e^{i\sigma x}$ cancels out with the spectral term $-\sigma^2 e^{i\sigma x} v(\bfx)$, and the next most important term	
	comes from one derivative in $-\partial_x^2 \in \triangle$ hitting $e^{i\sigma x}$ and one hitting $v$.\footnote{This heuristic comes from computing the ``b-decay order'' \cite{Me93} of each term; an unweighted b-operator has exactly as many factors of $\langle r \rangle$ as derivatives. For example, $\langle r \rangle\partial_x$ is an unweighted b-operator.  Regularity under repeated application of b-operators amounts to Kohn--Nirenberg symbol estimates. The b-decay order of a general operator is the deficit of factors of $\langle r \rangle$. The ``main term'' in the heuristic described above is the one with the least b-decay. In Melrose's terminology, it is the b-normal operator. 
    
    This heuristic works because b-decay order keeps track of how the operators in question augment the decay of symbols, like $v$.} 
	So, if we want $Pu_1\approx 0$, i.e.\ 
	\begin{equation*} 
		P ( (2i\sigma)^{-1} e^{i\sigma x}v(\bfx)) \approx - f_0,
	\end{equation*} 
	where $f_0= Pu_0$, 
    we should choose $v(\bfx)$ such that it solves 
	\begin{equation}
		\partial_x v(\bfx) \approx e^{-i\sigma x} f_0 , 
	\end{equation}
	where the `$\approx$' means we are keeping only those terms in the large-$r$ asymptotic expansion of $f_0$ preventing it from lying in the domain of the resolvent.  Integrating,
	\begin{equation} 
	 v(\bfx) = C(y,z) - \int_{-\infty}^x \frac{s a_1+ya_2+za_3}{(s^2+y^2+z^2)^{3/2}} \dd s . 
	\end{equation} 
	Since we want $v(\bfx)$ to be small near the backwards direction, we take $C(y,z)=0$. Inserting cutoffs to localize us near the backwards direction (since this was the only problematic location), we arrive at the formula \cref{eq:modified_ansatz}.
\end{remark*}

Explicitly, 
\begin{multline}
	v(x,y,z) = -\int_{-\infty}^x \frac{s a_1+y a_2+z a_3}{(s^2+y^2+z^2)^{3/2}}  \dd s \\= 
	\begin{cases}
			 \frac{a_1}{\sqrt{x^2+y^2+z^2}} -\frac{ya_2+za_3}{y^2+z^2}\Big[ 1+\frac{x}{\sqrt{x^2+y^2+z^2}}\Big] & (y^2+z^2\neq 0), \\
		a_1/|x| &(\text{otherwise}).
	\end{cases}
\end{multline}

We have $P u_1 = f_1$ for 
\begin{multline}
f_1 = \underbrace{{- \Big( \color{darkblue}\frac{\bfx\cdot \bfa}{r^3}} + W(\bfx)\Big)e^{i\sigma x}}_{f_0}  + \underbrace{\Big[P,\chi \Big(\frac{1}{r}\Big) \chi \Big( \frac{\sqrt{y^2+z^2}}{x}\Big) 1_{x\leq 0} \Big]}_{Q} \Big(\frac{e^{i \sigma x}}{2i\sigma}v(\bfx)\Big) \\ + \chi \Big(\frac{1}{r}\Big) \chi \Big( \frac{\sqrt{y^2+z^2}}{x}\Big) 1_{x\leq 0} \Big[\underbrace{ {\color{darkblue}\frac{\bfx\cdot \bfa}{r^3}} e^{i\sigma x} }_{\mathrm{I}}  +  \underbrace{\frac{e^{i\sigma x} }{2i\sigma}( P +\sigma^2 ) v}_{\mathrm{II}}    \Big].
\label{eq:misc_yyy}
\end{multline}
The term I results from one derivative hitting $v$ and one hitting $e^{i\sigma x}$. The term II contains the terms in which two derivatives hit $v$, together with the contribution from the subleading part of the potential.  The two most important terms in \cref{eq:misc_yyy}, the ones involving the dipole potential, have been highlighted. Note that they have \emph{opposite signs}, so they cancel (by design).

\begin{proposition}
    The term $f_1$ lies in the domain $\calY_+$ of the resolvent.
\end{proposition}
\begin{proof}
The sc-wavefront set of $f_1$ is contained in that of $e^{i\sigma x}$, which we already saw is disjoint from the incoming radial set $\calR$ except over the backwards direction. Hence it suffices to restrict attention to a small neighborhood of the backwards direction, say on 
\begin{equation}
    \Upsilon=\bigg\{\chi \Big(\frac{1}{r}\Big) \chi \Big( \frac{\sqrt{y^2+z^2}}{x}\Big)1_{x<0}=1\bigg\},
\end{equation}
where $Q$ vanishes identically. 
Then, 
\begin{equation}
    f_1 = - \Big(\frac{\bfx\cdot\bfa}{r^3}+W\Big)e^{i\sigma x} + \mathrm{I}+\mathrm{II} \text{ on } \Upsilon. 
\end{equation}
The term $W(\bfx)e^{i\sigma x}$ is $O(1/r^3)$ and therefore unproblematic. (Recall that $O(1/r^2)$ was the threshold we wanted to beat.) The other term in $f_0$, $-\bfx\cdot \bfa/r^3$, cancels out with $\mathrm{I}$ on $\Upsilon$, as discussed above.

Regarding (II), we note that \begin{equation} 
	( P+\sigma^2  )v =   \frac{\partial}{\partial x}\Big( \frac{\bfx\cdot \bfa}{r^3} \Big)  - \Big(  \frac{\partial^2}{\partial y^2}+ \frac{\partial^2}{\partial z^2} + \frac{\bfa\cdot \bfx}{r^3} + W \Big) v .
	\end{equation}
	The first term is $O(1/r^3)$, as is the term $( r^{-3} (\bfa\cdot \bfx)+W)v$, since $v=O(1/r)$. The terms involving $\partial_y,\partial_z$ are similar to each other, so we only discuss $\partial_y$. Explicitly, 
	\begin{equation}
	\frac{\partial v}{\partial y} = \int_{-\infty}^x \Big[-\frac{a_2}{(s^2+y^2+z^2)^{3/2}} +\frac{3y(s a_1+y a_2+z a_3)}{(s^2+y^2+z^2)^{5/2}}  \Big]\dd s , 
	\end{equation}
	which is $O(1/r^2)$. Taking another derivative yields something which is $O(1/r^3)$:
    \begin{equation}
	\frac{\partial^2 v}{\partial y^2} = \int_{-\infty}^x \Big[  \frac{3(s a_1+3y a_2+z a_3)}{(s^2+y^2+z^2)^{5/2}} - \frac{15y^2(s a_1+y a_2+z a_3)}{(s^2+y^2+z^2)^{7/2}}  \Big]\dd s .
	\end{equation}
\end{proof}

\begin{proposition}
    The conclusion of \Cref{prop:short_range_main} applies, except with $u_1$ in place of $u_0$ \cref{eq:Coulomb_ansatz} and $f_1$ in place of $f_0$. 
    \label{thm:Coulomb_main_refined}
\end{proposition}
\begin{proof}
    
    By
    \Cref{thm:main}.  
\end{proof}

It is straightforward to show that $f_1$ lies in the domain of the resolvent, $\calY_+$, uniformly as $\sigma\to\infty$, so we also get uniform estimates on the perturbed plane waves.



\section{Application: convergence of the Dyson series}
\label{sec:Dyson}

When solving the Schr\"odinger equation or wave equation on $\Omega=\bbR_t\times X$, one uses a propagator if the initial-value problem is of interest, or a Green function $G_\bullet \in \calS'(\Omega^2) $ if the forced equation is of interest. 
For simplicity, we begin by considering only the retarded (forward) problem for the forced wave equation,  
\begin{equation}
    (\partial_t^2 + \triangle )u +Vu =f. 
\end{equation}
Other problems, such as the advanced problem or the Cauchy problem, are similar. The Schr\"odinger equation is somewhat different, owing to infinite speed of propagation.  \emph{We assume, in this section, that $P=\triangle+V$ has no bound states.} (If we were to treat the Schr\"odinger equation rather than the wave equation, these states would not cause problems.)

The \emph{Dyson series} for the retarded Green function $G=G_V$ is the result of iterating the following exact relation between $G$ and the free retarded
Green function $G_0$:
\begin{equation}
	G(t,x,t',x') = G_0(t,x,t',x')-   \int_{\bbR\times X} G_0(t,x,t_0,x_0) V(x_0) G(t_0,x_0,t',x')   \dd t_0 \dd x_0 ,
    \label{eq:Dyson_0}
\end{equation} 
i.e.\ 
\begin{equation}
    G = G_0 (1-M_V G)
    \label{eq:Dyson_short}
\end{equation}
in shorthand, as results from applying $G_0$ to both sides of $(\partial_t^2 + \triangle)u = f - Vu$.
Formally, plugging \cref{eq:Dyson_0} back into itself and iterating (or, equivalently, solving for $G=(I+G_0 M_V)^{-1}G_0$ in \cref{eq:Dyson_short} and expanding) yields
\begin{equation}
G =  \sum_{j=0}^\infty (-G_0 M_V)^j G_0
\end{equation}
(identifying $G_0$, which is a Schwartz kernel, with the corresponding operator).
This is the Dyson series. 
What we would like to prove is that it converges in an appropriate sense. 
One difficulty is that the Green function receives contributions from low energy, which are out of reach of traditional microlocal methods. Consequently we extract the high-frequency component of the Dyson series, which is formally 
\begin{equation}
    \check{1}_{|\sigma|>\Sigma}(t) * G = \sum_{j=0}^\infty \check{1}_{|\sigma|>\Sigma}(t)*(-G_0 M_V)^j G_0.
    \label{eq:high_freq}
\end{equation}
Here, $\Sigma\gg 1$ is our threshold for what counts as ``high energy.'' Note that we do \textit{not} remove the low-energy component of every Green function in the series. Most are the full free Green functions, with no energy cutoff; the energy cutoff is applied only once per $j$. The convolution 
\begin{equation} 
    \check{1}_{|\sigma|>\Sigma}(t) * G \in \calS'(\bbR_t ; \calD'(X))
\end{equation} 
will be shown to be distributionally well-defined in the proof of \Cref{dysonconv} below.
Note that the energy cutoff 
\begin{align}
\begin{split} 
    \check{1}_{|\sigma|>\Sigma}(t) =  \int_{|\sigma|\geq \Sigma}e^{i \sigma t} \frac{\dd \sigma}{2\pi}  &= \frac{1}{\pi}\lim_{\varepsilon \to 0^+} \Big[ \int_{\sigma=\Sigma}^\infty e^{-\varepsilon \sigma} \cos(\sigma t) \dd \sigma \Big] \\ 
    &= \frac{i}{2\pi  } \Big[\frac{e^{i \Sigma t}}{  t + i 0}  - \frac{e^{-i \Sigma t}}{t-i0}  \Big]
    \end{split}
\end{align}
is the distributional inverse Fourier transform of the indicator function $1_{|\sigma|>\Sigma} \in \calS'(\bbR_\sigma)$. 

What we will prove is:
\begin{proposition}\label{dysonconv}
For every $k\in \bbN$ and $\chi\in C_{\mathrm{c}}^\infty(X^\circ)$ supported away from the marked points, there exists a $\Sigma>0$ such that
\begin{equation}
     M_\chi \check{1}_{|\sigma|>\Sigma}(t) * G M_{\chi} = \sum_{j=0}^\infty M_\chi  \check{1}_{|\sigma|>\Sigma}(t)*(-G_0 M_V)^j G_0 M_{\chi},
     \label{eq:Dyson_main}
\end{equation}
where the tail of the series is convergent as an operator $L^2(X)\to
C^k(X)$, uniformly in $t$.
\end{proposition}
One aspect of the proposition appears surprising: even though we have no control on $R_0(\sigma -i0)$ for $\sigma$ small, the series in \cref{eq:Dyson_main} converges \emph{with an energy cutoff applied only once per $j$}, on the left.
\begin{proof} 
In terms of the resolvent, $G$ can be written as:
\begin{equation}\label{eq:RtoG}
Gf(t,x) = \int_{-\infty}^\infty e^{i \sigma t} R_{V}(\sigma -i0) \hat{f}(\sigma,x) \frac{\dd \sigma}{2\pi}
\end{equation}
for $f\in C_{\mathrm{c}}^\infty(\bbR_t\times X)$ supported away from
the marked points. Here, $\smash{\hat{f}}$ is the temporal Fourier
transform of $f$.  This enables us to \emph{define} the high energy contribution $G^\Sigma f(t,x) = \check{1}_{|\sigma|>\Sigma}(t)*G$ as
\begin{equation*}
G^\Sigma f(t,x) \overset{\mathrm{def}}{=}\int_{|\sigma|>\Sigma} e^{i
  \sigma t} R_{V}(\sigma -i0) \hat{f}(\sigma,x) \frac{\dd \sigma}{2\pi }.
\end{equation*}
We remark that this definition tacitly relies upon resolvent estimates afforded by the limiting absorption principle:
since $R_{V}(\sigma-i0)$ is locally (in $\sigma$)  uniformly bounded
between Sobolev spaces as in \Cref{thm:main}, the resolvent is in fact regular enough to pair
with $1_{|\sigma|>\Sigma}$.

Plugging in the Born series now yields 
\begin{equation}\begin{aligned}
	G^\Sigma f(t) &= \sum_{j=0}^\infty \int_{|\sigma|>\Sigma}
        e^{i\sigma t} R_0(\sigma -i0) (-M_V R_0(\sigma-i0))^j
        \hat{f}(\sigma,-) \frac{\dd \sigma}{2\pi }\\
        &=
	\sum_{j=0}^\infty \big(\check{1}_{|\sigma|>\Sigma}* \calF^{-1}
 R_0(\sigma -i0) (-M_V R_0(\sigma-i0))^j
        \calF f(\sigma,-)\big)(t)
        \end{aligned}
	\label{eq:misc_555}
\end{equation}
(this converging absolutely, for each $f\in C_{\mathrm{c}}^\infty(\bbR_t\times K)$
for given $K \Subset X$ if $\Sigma=\Sigma(K)$ is large
enough\footnote{We are allowing $\Sigma$ to depend on $K$, so we are
not saying anything about the behavior of $G(t,x,t,x')$ as $x'\to\infty$.}).  Note that again the definition of the
convolution can be taken to be via the Fourier multiplication.

Since
  \begin{equation}
\calF_{t \to \sigma} (G_0 f(t,\bullet)) = R_0(\sigma-i0)(\calF_{t\to
  \sigma} f(t, \bullet))
\end{equation}
 (cf.\ \cref{eq:RtoG}), and since the spatial operator $M_V$ commutes
 with time Fourier transform (from which we now drop the subscripts),
 iteration yields
  \begin{equation}\label{eq:misc_096}
\calF (G_0(M_V G_0)^j f)(\sigma, \bullet) =
R_0(\sigma-i0)(M_VR_0(\sigma-i0))^j(\calF f)(\sigma, \bullet).
\end{equation}
Applying inverse Fourier transform and inserting this into
\cref{eq:misc_555} yields the desired expression.
\end{proof} 

The low-energy part of $G$, $\check{1}_{|\sigma|<\Sigma}(t)* G$, is 
\begin{equation}
    \check{1}_{|\sigma|<\Sigma}(t)* G = \int_{|\sigma|<\Sigma} e^{i\sigma t} R_{V}(\sigma-i0) \frac{\dd \sigma}{2\pi} . 
\end{equation}
This is a smoothing operator: it smooths completely in time and by two orders in space. 

\subsection[Example of smoothing behavior]{Example of smoothing behavior: $\protect X=\overline{\bbR^3}$}

The free retarded Green function on Minkowski spacetime (solving the \emph{forward} inhomogeneous problem) is simply 
\begin{equation}
    G_0(t,\bfx,t',\bfx') = \frac{1_{t>t'} }{2\pi}\delta( (t-t')^2 - |\bfx-\bfx'|^2 ). 
\end{equation}
That is, for $f\in C_{\mathrm{c}}^\infty(\bbR^{1,3})$, 
\begin{equation}
    G_0(t,\bfx) f = \frac{1}{4\pi} \int \frac{f(t_{\mathrm{ret}} , \bfx')}{ |\bfx-\bfx'|}   \dd^3 \bfx',\qquad  t_{\mathrm{ret}}=t-|\bfx-\bfx'|.
\end{equation} 
(As above, we abuse notation by confusing the Schwartz kernel with the
corresponding operator.)  The Schwartz kernel of this operator is
manifestly in no better a Sobolev space than $H^{-1/2-0}_{\loc}$ in
the spacetime region $t>0$.  In dealing with the perturbed
propagator $G$ we of course lack such explicit formulas, but can
nonetheless adopt the spectral-theoretic viewpoint; writing the propagator as
\begin{equation} 
G =1_{t>t'} \frac{\sin ((t-t')\sqrt{P})}{\sqrt{P}} \delta(\bfx-\bfx')
\end{equation} 
likewise gives us no higher regularity than $H^{-1/2-0}_{\loc}$,
and now without any a priori information about the location of singularities.  It maps $L^2_{\mathrm{c}}(\bbR^{1,3})$ to a space no better than $H^1_\loc$ in spacetime, as can be seen by unitarity of the evolution on the appropriate energy space, together with Duhamel's principle. 

By contrast, let us now examine the mapping properties of the $j$th term in the high-frequency Dyson series.  Fix $A>0$ and $F_0\Subset \RR^3$ compact.  Let $\tF =[0,A] \times F_0 \Subset \RR^{1,3}.$  Then for $f \in L^2(\tF)$,
\begin{equation} 
D_j(t)f := \chi  \check{1}_{|\sigma|>\Sigma}(t)*(-G_0 M_V)^j G_0 M_{\chi}f.
\end{equation}
By \eqref{eq:misc_096} and \Cref{thm:main}, for every $M$ there is a $J$ such that for $j \geq J$ this term is the inverse Fourier transform in time of a $O_{L^2_\loc}(\langle\sigma\rangle^{-M})$ function, i.e.\ (by the closed graph theorem),
\begin{equation} \label{Dj}
    D_j(t) \in \mathcal{L}(L^2(\tF),H^M(\RR_t; L^2_\loc(\bbR^3))),\quad j \geq J.
\end{equation} 
 The terms are thus \emph{smoother and smoother} as we go further out in the series.  Indeed, the same holds for the whole tail of the Dyson series as well, by the convergence results for the tail in \Cref{thm:main}.
 Note the perturbation-theoretic interpretation: these terms in the perturbation series account for successively more and more instances of free propagation interrupted by interaction with the potential singularity.  As is well known, this can give rise to regularizing (``diffractive'') effects even when the metric has a conic singularity (see, e.g., \cite{BaWu13}).  Here, the situation is even simpler, as we are able to treat the potential perturbatively.

 The discussion above takes care of the tail of the Dyson series. The initial terms in the Dyson series can be analyzed directly
 owing to the escape of singularities allowed to ricochet only a bounded number of times. 
  In what follows, we say that a subset of the characteristic set of the wave operator is outgoing \emph{relative to} $K\Subset \bbR^3$ if the bicharacteristic flowout from that set only intersects $\bbR_t\times K$ going backwards in time.
 \begin{lemma}
 Choose $\tF$ as above.
Fix $f\in L^2(\tF)$,  $j\in \bbN$.
\begin{enumerate}[label=(\alph*)]
    \item Given any $T>0$, the portion of the support of $(G_0 M_V)^j G_0 f$ in $(-\infty,T]\times \bbR^3$ is compact. 
    \item For any compact $K\Subset \bbR^3$, there exists a time $T=T_{\tilde{F},j,K}>0$  such that 
    \begin{equation} 
        (G_0 M_V)^j G_0 f\in C^\infty((T,\infty);L^2(K) ).
    \end{equation} 
    \item There exists a $T=T_{\tilde{F},j,K}>0$ such that, after time $T$, the portion of the wavefront set of $(G_0 M_V)^j G_0 f$ in the characteristic set of the free wave operator
    is entirely outgoing relative to $K$.
\end{enumerate}
\end{lemma}

\begin{proof}
We prove the desired results by induction on $j$. First consider $j=0$. We are thus studying $G_0 f$, the solution to the forward problem for the free wave equation with forcing $f$. Finite speed of propagation, combined with the fact that $G_0$ is the \emph{retarded} propagator, gives (a).
        Because $f$ is compactly supported, for any $K\Subset \bbR^3$ the solution  vanishes identically near $K$ after some late time. Thus, (b) holds, trivially. 
        The microlocal condition, (c), holds by virtue of the weak Huygens principle (i.e., escape of singularities).

         Now consider $(G_0 M_V)^j G_0 f$. If we already know that $(G_0 M_V)^{j-1} G_0 f$ satisfies (a), then it follows immediately from finite speed of propagation of support that $(G_0 M_V)^j G_0 f$ does as well.

        In order to verify (b): assuming 
        that $(G_0 M_V)^{j-1} G_0 f$ satisfies (b), (c), pick $T'>0$ large enough such that
        \begin{itemize}
            \item $
            (G_0 M_V)^{j-1} G_0 f \in C^\infty((T',\infty);L^2(K_0)) $, for some compact $K_0\Subset \bbR^3$ containing the various marked points and satisfying $K^\circ_0 \Supset K$, 
            \item the portion of the wavefront set of $(G_0 M_V)^{j-1} G_0 f$ in the characteristic set is entirely outgoing relative to $K$ for $t> T'$.
        \end{itemize}
        Now pick a partition of unity $1=\psi_{\mathrm{I}}(t)+\psi_{\mathrm{II}}(t)$ with $\psi_{\mathrm{I}}$ supported in $(-\infty,T'+1)$ and $\psi_{\mathrm{II}}$ supported in $(T', \infty)$. Let $\chi\in C_{\mathrm{c}}^\infty(\bbR^3)$ be identically $1$ near $K$ and the marked points, while being supported within $K_0^\circ$. Writing 
        \begin{multline}
            (G_0 M_V)^{j} G_0 f = \underbrace{(G_0 M_V) \psi_{\mathrm{I}} (G_0 M_V)^{j-1} G_0 f}_{\mathrm{I}}+(\underbrace{G_0 M_V) \psi_{\mathrm{II}} \chi (G_0 M_V)^{j-1} G_0 f}_{\mathrm{II}} \\ 
            +(\underbrace{G_0 M_V) \psi_{\mathrm{II}} (1-\chi) (G_0 M_V)^{j-1} G_0 f}_{\mathrm{III}} ,
            \label{eq:terms}
        \end{multline}
        we show that each term on the right-hand side has the desired late-time smoothness. 
        \begin{enumerate}[label=(\Roman*)]
            \item By part (a) for $j-1$ (i.e., part of our inductive hypotheses), $M_V\psi_{\mathrm{I}} (G_0 M_V)^{j-1} G_0 f$ has compact support in spacetime. So, by Huygens, term I vanishes identically on $K$ at late enough times.
            \item By the Hardy inequality, $M_V:H^1(\bbR^3)\to L^2(\bbR^3)$. Dualizing, $M_V:L^2(\bbR^3)\to H^{-1}(\bbR^3)$. Thus, using the inductive hypothesis,
            \begin{equation} 
            M_V \psi_{\mathrm{II}} \chi (G_0 M_V)^{j-1} G_0 f \in C^\infty(\bbR_t; H^{-1}(\bbR^3)) .
        \end{equation}
        The free propagator smooths by one order, meaning 
        \begin{equation}
            G_0 : C^\infty(\bbR_t; H^{-1}(\bbR^3)) \to C^\infty(\bbR_t; L^2(\bbR^3)),
        \end{equation}
        as can be concluded e.g.\ by Duhamel.
        Applying this to $ M_V \psi_{\mathrm{II}} \chi (G_0 M_V)^{j-1} G_0 f$, we conclude that term II is in $C^\infty(\bbR_t; L^2(\bbR^3))$. 
        \item The final term is controlled using the inductive hypothesis that $(G_0 M_V)^{j-1} G_0 f$ satisfies (c). Notice that there exists a \emph{smooth} function $\tilde{V}$ that agrees with $V$ on the support of $1-\chi$. Hence term III is 
        \begin{equation}
            G_0 (M_{\tilde{V}} \psi_{\mathrm{II}}(1-\chi) (G_0 M_V)^{j-1} G_0 f). 
            \label{eq:misc_112}
        \end{equation}
        By construction, the portion of the wavefront set of $ \psi_{\mathrm{II}}(1-\chi) (G_0 M_V)^{j-1} G_0 f$ in the characteristic set of the free wave operator is entirely outgoing.  Note that $M_{\tilde{V}} \psi_{\mathrm{II}}(1-\chi) (G_0 M_V)^{j-1} G_0 f$ vanishes near $K$ (by the support condition on $1-\chi$).  Microlocal elliptic regularity implies that
        \begin{equation}
        \WF (G_0 (M_{\tilde{V}} \psi_{\mathrm{II}}(1-\chi) (G_0
        M_V)^{j-1} G_0 f)) \cap T^* (\RR_t \times K)
        \end{equation} lies inside the characteristic set, and
        propagation of singularities then shows that this set is in
        fact empty, since $\WF (M_{\tilde{V}}
        \psi_{\mathrm{II}}(1-\chi) (G_0 M_V)^{j-1} G_0 f)$ is entirely
        $K$-outgoing.  
        \end{enumerate}
        This establishes (b).
        
        Finally, we verify (c), term-by-term in \cref{eq:terms}. Term I certainly satisfies (c). Since term II is in $C^\infty(\bbR_t;L^2(\bbR^3))$, its wavefront set must lie in the characteristic set of $D_t$ (by microlocal elliptic regularity). The characteristic set of $D_t$ is disjoint from the characteristic set of the wave operator. Finally, it follows from \cref{eq:misc_112} that the wavefront set of term III in the characteristic set of the wave operator propagates outwards, so stays $K$-outgoing at late times. 
\end{proof}

\begin{corollaryx}\label{cor:Huygens}
Given compact sets $\tF\subset \RR^{1,3}$ and $K\subset \RR^3$ and $M\in\mathbb{N}$, there exists a $T$ such that,
 \begin{equation} 
    G \in \mathcal{L}(L^2(\tF),H^M_{\mathrm{loc}}((T, \infty); L^2(K))).
\end{equation}
\end{corollaryx}

\begin{proof} For any $f$, $Gf$ is a sum of the smooth part at low frequency, the smooth part from finitely many terms in the series, and the tail of the series, which has the asserted regularity by \eqref{Dj} et seq.  The estimate as stated (rather than for a single $f$ at a time) then follows by the closed graph theorem.\end{proof}

This ``very weak Huygens principle,'' 
in which solutions ultimately achieve any desired degree of smoothness after a sufficiently long time,
was used in \cite{BaWu13} to show the existence of logarithmic resonance-free regions in the complex plane. (See \cite[Chapter 4]{DyZw19} for the definition of resonances, which are most easily viewed as poles of the analytic continuation of the resolvent, mapping from $L^2_{\mathrm{c}}(\RR^3)$ to $L^2_{\text{loc}}(\RR^3)$.)
The proof followed from a variant of the Vainberg parametrix construction \cite{Va88}.
As the parametrix step of \cite{BaWu13} was obtained in the full generality of black box scattering, the same proof applies here, provided we confine ourselves to fully screened (i.e., compactly supported) potentials with Coulomb singularities.

\begin{corollaryx}
Assume that $V$ is a real-valued compactly supported potential with
Coulomb singularities (i.e., satisfying \eqref{eq:Coulomb}) such
  that $\triangle +V$ has no bound states.  Then there exist $\nu,\ R>0$ such that
$\triangle +V$
has no resonances in the region
\begin{equation}\lvert \Re z\rvert>R,\quad \Im z > -\nu \log \lvert \Re z \rvert,\end{equation}
and in this region enjoys the cutoff resolvent estimate
\begin{equation}
\|\chi(\triangle+V-z^2)^{-1}\chi\|_{L^2\to L^2}\leq C \lvert z\rvert^{-1}e^{T\lvert \Im z \rvert}
\end{equation}
for all $\chi\in C_c^\infty (\RR^3)$ and for some $C,T$ depending on $\chi$.
\end{corollaryx}
\begin{proof}
The result follows from Proposition~8 of \cite{BaWu13}, together with
\Cref{cor:Huygens}, as follows.  Here (in the notation of
\cite{BaWu13}) the ``black box'' operator is given by the Hamiltonian
$P=\triangle+V$ acting on $L^2(\RR^3)$ with domain
$\mathcal{D}=H^2(\RR^3)$.  The black box Hilbert space
$\mathcal{H}_{R_0}$ is simply given by $L^2(B(0, R_0))$, with $R_0$
chosen large enough so that $\supp V$ is inside $B(0, R_0)$.  The
compactness of\footnote{We remark that there is a typographical error in \cite{BaWu13}: the factor of $1_{B(0, R_0)}$ is missing from the
conditions in Section 6.1; cf.~\cite[(4.1.12)]{DyZw19}.} $1_{B(0,
R_0)}(P+i)^{-1}$ follows by self-adjointness of $P$ with domain
$H^2$, together with the Rellich lemma.  The black box
hypotheses also require a Weyl-type upper
bound of the form \cite[Equation (6)]{BaWu13} for the operator
$P^\sharp$ obtained by compactifying onto a large reference
torus.  This follows from the Hardy inequality, which guarantees that for
any $\epsilon>0$,
$P^\sharp=\triangle+V \geq (1-\epsilon) \triangle-C_\epsilon$ on the
torus; the min-max characterization of the eigenvalues then gives a
suitable upper bound on the counting function.

Let $U(t)=\sin (t\sqrt{P})/\sqrt{P}$ denote the propagator mapping a function $g$ to the solution $u$ of the wave equation with Cauchy data $u(0)=0$, $u_t(0)=g$.  Thus, for $g \in L^2(\RR^3)$, $u \in L^\infty_\loc(\RR; H^1(\RR^3))$.  Fix $\psi(t)$ supported in $(0,\infty)$ and equal to $1$ on $(1,\infty)$.  If $g$ is supported in a compact set $F_0$,
\begin{equation}
(\pa_t^2+P)\psi(t) U(t)g=[\pa_t^2, \psi] U(t) g \in L^2([0,1]\times F_1)
\end{equation}
with $F_1$ containing a neighborhood of radius $1$ around $F_0$.   Thus by \Cref{cor:Huygens}, for any $K$ compact and $M \in \NN$, there exists $T$ such that 
\begin{equation}
U(t)g \in H^M_\loc((T,\infty); L^2(K)).
\end{equation}
More generally, by the Sobolev embedding and the closed graph theorem we find that for any $\chi \in C_c^\infty(\RR^3)$ and $s \in \NN$, there exists $T>0$ such that
\begin{equation}\label{regularizing}
\chi U(t) \chi\in C^s((T,\infty); \mathcal{L}( L^2,  L^2)).
\end{equation}
This is the regularizing effect necessary for the application of Proposition~8 of \cite{BaWu13}.\footnote{In that paper, the regularizing effect applied both in the spatial and temporal variables, and this was reflected in the theorem statements. Here, we can convert temporal regularity to get additional spatial regularity, but this is actually unnecessary, because \cref{regularizing} is all that is used  in \cite[Prop.\ 8]{BaWu13} and Vainberg's method more generally. Cf.\ Remark~2 following Theorem~4.43 of \cite{DyZw19}.}
\end{proof}

\appendix

\section{Various Sobolev spaces}
\label{sec:Sobolevs}
In this section we clarify the definitions of, and relationships among, the menagerie of Sobolev spaces employed above.  Each hierarchy of Sobolev spaces is associated to one or more pseudodifferential calculi. 
Our most refined calculus is the $\vee$- (``vee'') calculus, described at length below. The other calculi all arise from ignoring various aspects of the $\vee$-calculus. For example, one can restrict attention to $r\gg 1$, where the sc-calculus is the relevant one, or to $r\ll 1$, where the b-calculus is the relevant one (for individual $\sigma$). The $\vee$-spaces keep track of both, but the large- and small-$r$ problems are, in some sense, decoupled. Add to the mix the possibility of tracking $\sigma\to\infty$ behavior, and one gets more calculi still. If the reader prefers, they may use the $\vee$-Sobolev spaces throughout this paper, instead of the various ``forgetful'' spaces. The $\vee$-spaces originate in the work of Hintz \cite{Hintz2}, but Hintz did not give them a name.

We only describe spaces with a positive integer amount of differential order $m$.\footnote{Defining the spaces with fractional or negative values requires either interpolation and duality arguments, or the deployment of the relevant pseudodifferential calculi. We do not describe the pseudodifferential calculi here, apart from remarking that they microlocalize the vector fields used to define the Sobolev spaces. 
}
Membership in one such space means regularity under $m$-fold application of certain vector fields and multiplication by some polynomial weights.
The spaces are thus specified by the vector fields and weights. 
The norms are left implicit, but can be written as square roots of sums of squares of all the quantities whose finiteness determines the space.  To ease our description of the spaces we stick to the case $X=\overline{\RR^n}$, but the definitions can be extended to general scattering manifolds as described in \cite{Me94} straightforwardly, using local coordinate charts.
 
 In most of the spaces below there is a semiclassical parameter $h$ (which will equal $\sigma^{-1}$ in our application), and the norms are correspondingly $h$-dependent.  
 If $\calX=\{\calX(h)\}_{h>0}$ is a family of normed function spaces $\calX(h)\subset \calD'$ and $u=\{u(-;h)\}_{h>0}$ is a family of distributions, then we say that $u$ lies in $\calX$ \emph{uniformly} if $u(-;h)\in \calX(h)$ for each $h>0$ and $\lVert u(-;h)\rVert_{\calX(h)}$ is uniformly bounded as $h\to 0^+$.

\vspace{1em}
\noindent \textbf{The semiclassical calculus} In the interior of $X$ away from marked points (or ignoring marking), a one-parameter family $u=\{u(-;h)\}_{h>0}$ of distributions $u(-;h)$ lies uniformly in $H_\hbar^m(X)$ if it lies uniformly in $L^2$ and up to $m$-fold products of vector fields $h D_j$ (in local coordinates) leave it uniformly in $L^2$.  See e.g.\ \cite[App.~E]{DyZw19}.

\vspace{1em}
\noindent \textbf{The semiclassical scattering calculus} Away from marked points (or ignoring marking), a one-parameter family $u=\{u(-;h)\}$ of distributions $u(-;h)$ is uniformly in $H_{\sc,\hbar}^{m,s}(X)$ if $r^s u$ lies uniformly in $L^2(X)$ as do up to $m$-fold products of vector fields $h D_j$ applied to $r^s u$.  (This of course looks the same as the space above with a weight added, but we are now working globally.)  The differential order $m$ here is of little interest in our work, as our operator $P$ is elliptic at fiber infinity.

More subtly, we may replace the weight $s$ with a \emph{function} $\mathsf{s}$ on the phase space, which here is the \emph{compactified scattering cotangent} bundle, identified in the case of $X=\overline{\RR^n}$ with the space $\overline{\RR^n_x} \times \overline{\RR^n_\xi}$. The weight $r^s$ now becomes a (scattering) pseudodifferential operator $\Op (r^\mathsf{s})$, having variable order.   See e.g.\ \cite{VaZw00}.

When only one value of $h$ is of interest, the `$\hbar$' will be dropped from the notation.

\vspace{1em}
\noindent 
\textbf{The semiclassical b-calculus} Near a marked point, which without loss of generality we take to be at $r=0$, a one-parameter family $u=\{u(-;h)\}_{h>0}$ of distributions $u(-;h)$  is in $H_{\mathrm{b},\hbar}^{m,\ell}([X; \text{markings}])$ uniformly if $r^{-\ell} u$ lies in $L^2$ as do up to $m$-fold products of vector fields  $h r D_r$, $h D_{\theta_j}$ (in polar coordinates) applied to $r^{-\ell} u$. 

When only one value of $h$ is of interest, the `$\hbar$' will be dropped from the notation.  Then, the norm is the ordinary b- Sobolev norm.

\vspace{1em}
\noindent \textbf{The $\mathrm{sc}-\mathrm{b}$ calculus} A distribution $u$ is in $\smash{H_{\scb}^{m,\mathsf{s}, \ell}}([X; \text{markings}])$ if near marked points it lies in $H_{\mathrm{b}}^{m,\ell}$ and near $\pa X$ it lies in $H_{\sc}^{m,\mathsf{s}}$.  

Thus, near marked points the norm is specified by powers of $rD_r$, $D_\theta$ applied to $r^{-\ell} u$ and near $\pa X$ by powers of $D_j$ applied to $\Op (r^\mathsf{s})u$.  This is the only one of the calculi under discussion in which we have no semiclassical parameter, and we use it to write the simplest possible qualitative results. See e.g.\ \cite[Appendix A]{HiVa18}.\footnote{The space $H_{\mathrm{sc-b}}(X)$
  should not be confused with what Vasy calls $H_{\mathrm{sc,b}}(X)$
  in \cite{VasyOverloaded}.}

\vspace{1em}
\noindent \textbf{The ``vee'' calculus} This calculus and its associated Sobolev space derive from the semiclassical cone calculus of Hintz \cite{Hintz2}: $u$ is in $H_\vee^{m,\mathsf{s}, \ell,\alpha}([X;\text{markings}])$ if 
\begin{itemize}
    \item 
$u \in H_{\mathrm{sc},\hbar}^{m,\mathsf{s}}$ near $\pa X$, uniformly, while 
    \item near the marked points, 
    \begin{equation*} 
        \Big( \frac{r}{r+h}\Big)^{-\ell} (r+h)^{-\alpha} u
    \end{equation*}
    lies in $L^2$ uniformly, as do up to $m$-fold products of vector fields $\frac{h}{h+r} r D_r$, $\frac{h}{h+r} D_\theta$ (in polar coordinates), applied to it.
\end{itemize}  
Note that for fixed $h>0$ the norms on this space are equivalent to those on $H_{\scb}^{m,\mathsf{s},\ell}$, i.e., membership in the two spaces is the same; it is only as $h \downarrow 0$ that the norms on $H_\vee^{\bullet}$ are varying and thus encoding information about asymptotics in the parameter $h$.

If one is willing to give up a polynomial in $h$, it is possible to convert between sc-b and $\vee$-norms:
\begin{equation}\label{eq:sobolevinclusion1}
\| u \|_{H_\vee^{m,s, \ell,\alpha}}<C\Longrightarrow 
\| u \|_{H_\scb^{m,s, \ell}}<C' h^{-\lvert m\rvert -\lvert \alpha\rvert -\lvert \ell \rvert-|s|}
\end{equation}
since
\begin{equation} 
\|(r D_r)^\alpha  D_\theta^\beta u\| \lesssim \sum_{\alpha'\leq \alpha,\beta'\leq \beta} h^{-\alpha'-\lvert\beta'\rvert}
\Big\| \Big(\frac{h}{h+r} r D_r)^{\alpha'} \Big(\frac{h}{h+r} D_\theta\Big)^{\beta'} u\Big\|,
\end{equation} 
which establishes the estimate near the marked points, and likewise
\begin{equation} 
\| D^\alpha u\| = h^{-\lvert \alpha \rvert} \|(h D)^\alpha u\|
\end{equation} 
handles near the boundary at infinity.

Conversely, we also have more immediately the opposite inclusion:
\begin{equation}\label{eq:sobolevinclusion2}
\| u \|_{H_\scb^{m,\mathsf{s}, \ell}}<C
\Longrightarrow
\| u \|_{H_\vee^{m,\mathsf{s}, \ell,\ell}}<C.
\end{equation}

\section[A bound on Kummer's confluent hypergeometric function]{An $L^\infty$ bound on Kummer's confluent hypergeometric function}

\begin{lemma}
	\label{lem:Kummer_bound}
    For $A>0$, Kummer's hypergeometric function $M$ satisfies 
    \begin{equation} 
        M(a,1,i\lambda) \in L^\infty(  i[-A,A]_a\times \bbR_\lambda).
    \end{equation}
\end{lemma}

\begin{proof}
    We can make use of the integral formula 
    \begin{equation} 
    M(a,1,z) = \frac{1}{\Gamma(1-a)\Gamma(a)} \int_0^1 e^{zt} t^{a-1} (1-t)^{-a} \dd t
    \label{eq:Kummer_integral}
    \end{equation}
    \cite[\href{http://dlmf.nist.gov/13.4.E1}{Eq.\ 13.4.1}]{NIST}, which holds for $a$ in the strip $\Re a \in (0,1)$. Unfortunately, we need the case $a\in i\bbR$, for which $\Re a=0$. Our first task is to derive from \cref{eq:Kummer_integral} a similar integral formula that holds for a wider range of $a$.

    Expand:
    \begin{equation}
        e^{zt}(1-t)^{-a} = 1 + E(t,a,z) , 
    \end{equation}
    where $E(t,a,z) = -1+ e^{zt}(1-t)^{-a}$; note that, for each $a\in \bbC$ and $z\in \bbC$, this is $O(t)$ as $t\to 0^+$. 
    Consequently, for $\Re a \in (0,1)$, 
    \begin{align}
        \begin{split} 
        \int_0^1 e^{zt} t^{a-1} (1-t)^{-a}\dd t &= \int_0^1 t^{a-1} \dd t + \int_0^1 t^{a-1}  E(t,a,z)\dd t  \\
        &= \frac{1}{a} + \int_0^1 t^{a-1}  E(t,a,z)\dd t ,
        \end{split} 
    \end{align}
    where the final integral is absolutely convergent for every $z\in \bbC$ and $a\in \bbC$ with $-1< \Re a < 1$, and evidently depends analytically on $a$. 
    Using the meromorphy of $M(a,1,z)$, it follows that
    \begin{equation}
        M(a,1,z) = \frac{1}{\Gamma(1-a)\underbrace{\Gamma(a) a}_{\Gamma(1+a)}} \Big[1+   a\int_0^1 t^{a-1}  E(t,a,z)\dd t\Big] 
        \label{eq:misc_099}
    \end{equation}
    whenever $-1 < \Re a < 1$.  

    The remaining task is to check that the integral above can be bounded uniformly in $a\in i[-A,A]$ and $z=i\lambda$, for $\lambda\in \bbR$. For bounded $\lambda$, this is straightforward, so suppose  $|\lambda|\geq 1$. We now use
    \begin{align}
        \begin{split} 
        -E(t,a,i\lambda) &= (1-e^{it\lambda})+(1-(1-t)^{-a})- (1-e^{i t \lambda })(1-(1-t)^{-a}) \\
        &=  (1-e^{it\lambda}) + O( t  ),
        \end{split} 
    \end{align}
    where the $O(t)$ bound is uniform.  The $O( t )$ term contributes a uniformly bounded term to $M(a,1,i\lambda)$. The other term contributes minus
    \begin{equation}
        a\int_0^1 t^{a-1} (1-e^{it \lambda}) \dd t = a\underbrace{\int_0^{1/|\lambda|} t^{a-1} (1-e^{it \lambda}) \dd t}_{\mathrm{I}} + \underbrace{\int_{1/|\lambda|}^1  at^{a-1} \dd t}_{\mathrm{II}} - a\underbrace{\int_{1/|\lambda|}^1 t^{a-1} e^{it\lambda}\dd t}_{\mathrm{III}}.
    \end{equation}
    Term I is estimated using $1-e^{it\lambda} = O(t|\lambda|)$, which yields
    \begin{equation}
        |\mathrm{I}| \lesssim |\lambda| \int_0^{1/|\lambda|} \dd t =1
    \end{equation}
    (using $\Re a=0$).
    Term II is just calculated:
    \begin{equation}
        \int_{1/|\lambda|}^1  at^{a-1} \dd t = 1-|\lambda|^{-a} = O(1)
    \end{equation}
    (again using $\Re a=0$).
    Term III is estimated by integrating by parts (to make use of the oscillations of the integrand  -- the only reason we are doing this is that it beats the na\"ive estimate $\mathrm{III}\lesssim \int_{1/|\lambda|}^1\dd t/t = \log |\lambda|$ by a log): for $\lambda\neq 0$, 
    \begin{equation}
        -\mathrm{III} =  \frac{i}{\lambda}\int_{1/|\lambda|}^1 t^{a-1} \frac{\mathrm{d}}{\mathrm{d}t}   e^{it\lambda} \dd t  = \underbrace{\frac{i}{\lambda} \Big[t^{a-1} e^{it\lambda} \Big]_{t=1/|\lambda|}^1}_{=O(1)} +\frac{a-1}{\lambda i}\int_{1/|\lambda|}^1  t^{a-2} e^{it\lambda} \dd t, 
    \end{equation}
    and
    \begin{equation}
        \Big|  \int_{1/|\lambda|}^1  t^{a-2} e^{it\lambda} \dd t\Big| \leq \int_{1/|\lambda|}^1 \frac{\dd t}{t^2} = |\lambda| -1.
    \end{equation}
\end{proof}

\begin{figure}[t]
    \centering
\includegraphics[scale=.55]{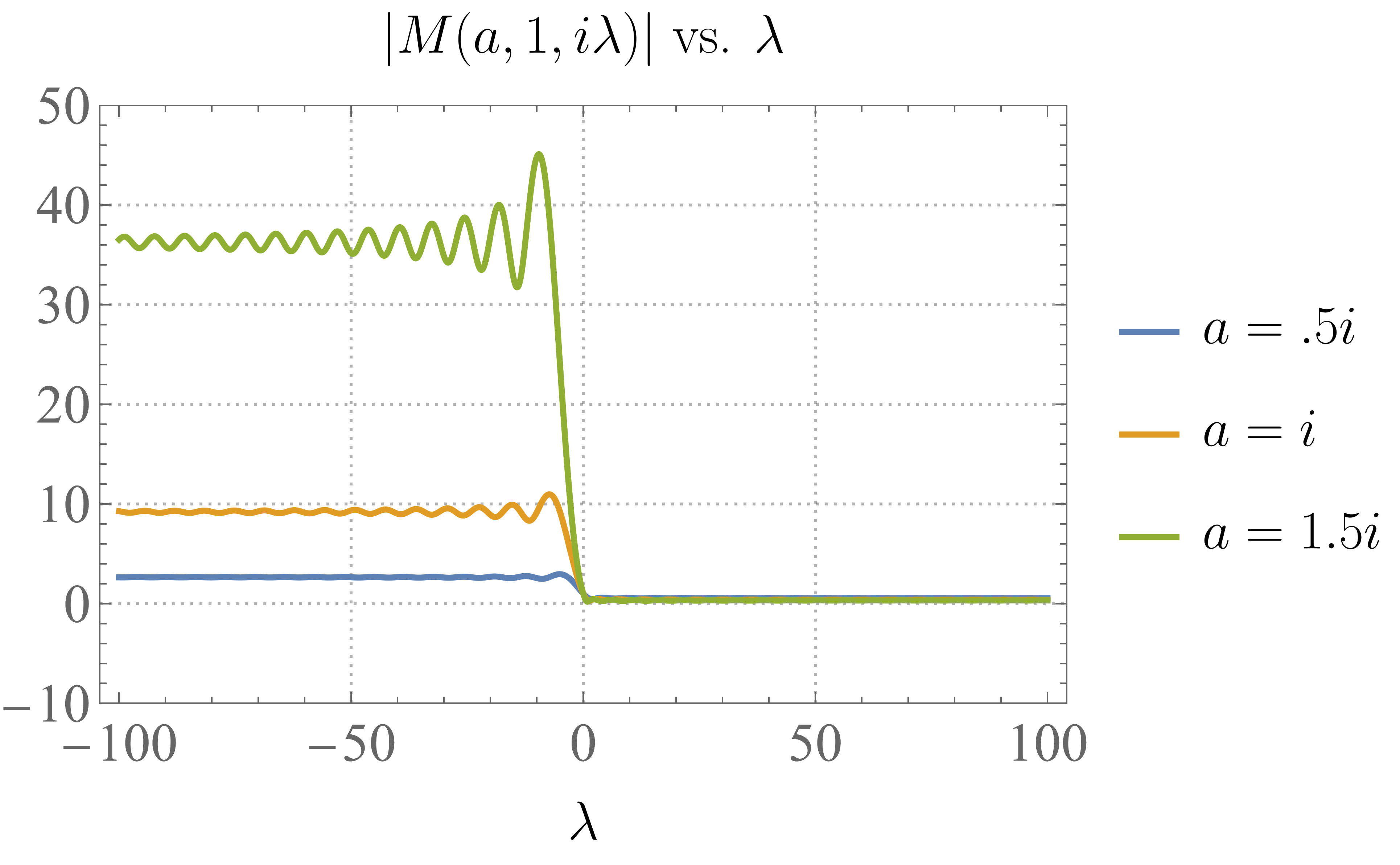}

    \caption{A plot of the absolute value of $M(a,1,i\lambda)$ versus $\lambda\in \bbR$, for $a\in i \bbR$. \Cref{lem:Kummer_bound} says that, for each individual $a$, we have an $L^\infty(\bbR_\lambda)$ bound on $M(a,1,i\lambda)$, with uniformity in $a$ as long as $\Im a$ stays bounded. In \S\ref{sec:multi-Coulomb}, we care about the $a\to \pm i0$ limit. What we want to emphasize here is that there exist $L^\infty(\bbR_\lambda)$-bounds uniform in this limit.}
    \label{fig:placeholder}
\end{figure}

When $a=0$, the $M$-function satisfies $M(a,1,i\lambda)=1$ identically \cite[\href{http://dlmf.nist.gov/13.2.E2}{Eq. 13.2.2}]{NIST}. It is thus unsurprising that the $L^\infty$-bounds that hold for $a\neq 0$ apply uniformly as $a\to 0$.

\printbibliography

\end{document}